%% file: fcf.tex
\documentclass[amsmath,amssymb,aps,twocolumn,nofootinbib,notitlepage,superscriptaddress]{revtex4-1}%
\usepackage{hyperref}
\usepackage{graphicx}
\usepackage[table]{xcolor}
\usepackage{latexsym}
\usepackage[normalem]{ulem}

\providecommand{\U}[1]{\protect\rule{.1in}{.1in}}
\newtheorem{theorem}{Theorem}

\newtheorem{proposition}[theorem]{Proposition}

\newenvironment{proof}[1][Proof]{\noindent\textbf{#1.} }{\ \rule{0.5em}{0.5em}}

\newcommand{\bit}{\begin{itemize}}
\newcommand{\eit}{\end{itemize}\par\noindent}
\newcommand{\ben}{\begin{enumerate}}
\newcommand{\een}{\end{enumerate}\par\noindent}
\newcommand{\beq}{\begin{equation}}
\newcommand{\eeq}{\end{equation}\par\noindent}
\newcommand{\beqa}{\begin{eqnarray}}
\newcommand{\eeqa}{\end{eqnarray}\par\noindent}
\newcommand{\beqn}{\begin{eqnarray}}
\newcommand{\eeqn}{\end{eqnarray}\par\noindent}

\newcommand{\idn}{\mathbb{I}}
\makeatletter
\def\hlinewd#1{%
  \noalign{\ifnum0=`}\fi\hrule \@height #1 \futurelet
   \reserved@a\@xhline}
\makeatother

\newcommand{\cit}[1]{~\cite{#1}}
\newcommand{\suppref}[1]{Appendix~\ref{supp#1}}

\begin{document}

\title{An experimental test of noncontextuality without unwarranted idealizations}

\author{Michael D. Mazurek}
\affiliation{Institute for Quantum Computing, University of Waterloo, Waterloo, Ontario N2L 3G1, Canada}
\affiliation{Department of Physics \& Astronomy, University of Waterloo, Waterloo, Ontario N2L 3G1, Canada}
\author{Matthew F. Pusey}
\affiliation{Perimeter Institute for Theoretical Physics, 31 Caroline Street North, Waterloo, Ontario N2L 2Y5, Canada}
\author{Ravi Kunjwal}
\affiliation{Optics \& Quantum Information Group, The Institute of Mathematical Sciences, C.I.T Campus, Taramani, Chennai 600 113, India}
\author{Kevin J. Resch}
\affiliation{Institute for Quantum Computing, University of Waterloo, Waterloo, Ontario N2L 3G1, Canada}
\affiliation{Department of Physics \& Astronomy, University of Waterloo, Waterloo, Ontario N2L 3G1, Canada}
\author{Robert W. Spekkens}
\affiliation{Perimeter Institute for Theoretical Physics, 31 Caroline Street North, Waterloo, Ontario N2L 2Y5, Canada}

\begin{abstract} %

To make precise the sense in which nature fails to respect classical physics, one requires a formal notion of classicality.  Ideally, such a notion should be defined operationally, so that it can be subjected to a direct experimental test, and it should be applicable in a wide variety of experimental scenarios, so that it can cover the breadth of phenomena that are thought to defy classical understanding.
Bell's notion of local causality fulfills the first criterion but not the second.  The notion of noncontextuality fulfills the second criterion, but it is a long-standing question whether it can be made to fulfill the first. Previous attempts to experimentally test noncontextuality have all presumed certain idealizations that do not hold in real experiments, namely, noiseless measurements and exact operational equivalences. We here show how to devise tests that are free of these idealizations.  We also perform a photonic implementation of one such test that rules out noncontextual models with high confidence.
\end{abstract}

\maketitle
\input{theory}
\input{pseudotheory}
\input{exp}
\begin{acknowledgments}
 The authors thank Megan Agnew for assistance with data acquisition software.  This research was supported in part by the Natural Sciences and Engineering Research Council of Canada (NSERC), Canada Research Chairs, Ontario Centres of Excellence, Industry Canada, and the Canada Foundation for Innovation (CFI).  MDM acknowledges support from the Ontario Ministry of Training, Colleges, and Universities. RK thanks the Perimeter Institute for hospitality during his visit there, which was made possible in part through the support of a grant from the John Templeton Foundation and through the support of the Institute of Mathematical Sciences, Chennai.  Research at Perimeter Institute is supported by the Government of Canada through Industry Canada and by the Province of Ontario through the Ministry of Research and Innovation.
\end{acknowledgments}
\bibliography{bibl}

\appendix
\setcounter{section}{0}

\input{supp}
\end{document}

%% file: theory.tex
\section{Introduction}
Making precise the manner in which a quantum world differs from a classical one 
is a surprisingly difficult task.  
The most successful attempt, due to Bell\cit{bell64}, shows a conflict between quantum theory and
a feature of classical theories termed {\em local causality}, which asserts 
that no causal influences propagate faster than light. But the latter assumption can only be tested for scenarios wherein there are two or more systems that are space-like separated.
And yet few believe 
that this highly specialized situation is the only point where the quantum departs from the classical.
A leading candidate for a notion of nonclassicality with a broader scope
 is the failure of quantum theory to admit of a noncontextual model, as proven by Kochen and Specker\cit{kochen68}.  
Recent work has highlighted how this notion
 lies at the heart of many phenomena that are taken to be distinctly quantum: the fact that quasi-probability representations go negative\cit{spekkens2008negativity,ferrie2008frame}, the existence of quantum advantages for cryptography\cit{spekkens09} and for computation\cit{Howard2014,raussendorf2013,hoban2011}, and the possibility of anomalous weak values\cit{Pusey2014}. 

An experimental refutation of noncontextuality 
would demonstrate that the conflict with noncontextual models is not only a feature of quantum theory, but of nature itself, and hence also of any successor to quantum theory.
The requirements for such an experimental test, however, 
have been a subject of much controversy\cit{meyer99finite,kent99noncontextual,cliftonkent, mermin1999kochen, simon01hidden, Larsson2002, barrett04}.  

A fundamental problem with most proposals for testing noncontextuality\cit{cabello98proposed,cabello00kochen,simon00feasible,cabello08proposed,cabello08experimentally,
badziag09universality,Guhne10compatibility},
 and experiments performed to date\cit{michler00,huang03,hasegawa06,Gadway2009,bartosik09,kirchmair09,amselem09,moussa10,lapkiewicz11},
 is that they assume that measurements have a deterministic response in the noncontextual model.  It has been shown that this can only be justified under the idealization that measurements are noiseless\cit{spekkens14}, which is never satisfied precisely by any real experiment.  
We here show how to contend with such noise. 

Another critical problem with previous proposals is the fact that the assumption of noncontextuality can only be brought to bear when two measurement events (an event is a measurement and an outcome) are {\em operationally equivalent}, which occurs when the two events 
are assigned exactly the same probability by all preparation procedures\cit{spekkens05}; in this case they are said to differ only by the measurement {\em context}.  In a real experiment, however, one never achieves the ideal of precise operational equivalence. 
 Previous work on testing noncontextuality---including the only experiment to have circumvented the problem of noisy measurements (by focusing on preparations)\cit{spekkens09}---has failed to provide a satisfactory account of how   the deviation from strict operational equivalence  should be accounted for in the interpretation of the results. 
 We here demonstrate a general technique that allows one to circumvent this problem.
 
For Bell's notion of local causality, the theoretical work of Clauser \emph{et.\ al.}\cit{clauser69} was critical to enabling an experimental test without unwarranted idealizations, such as the perfect correlations presumed in Bell's original proof\cit{bell64}. Similarly, the theoretical innovations we introduce here make it possible for the first time to subject noncontextuality to an
experimental test without the idealizations described above. 
We report on a quantum-optical experiment of this kind, the results of which rule out noncontextual models with high confidence.

\section{A noncontexuality inequality}
According to the operational approach proposed in ref.~\onlinecite{spekkens05}, to assume noncontextuality is to assume a constraint on model-construction, namely, that {\em if procedures are statistically equivalent at the operational level then they ought to be statistically equivalent in the underlying model}.
Operationally, a system is associated with a set $\mathcal{M}$ (resp.\ $\mathcal{P}$) of physically possible measurement (resp. preparation) procedures.  An {\em operational theory} specifies the possibilities for the conditional probabilities $\{ p(X|P,M): P\in \mathcal{P},M\in\mathcal{M}\}$ where $X$ ranges over the outcomes of measurement $M$.  In an {\em ontological model} of such a theory,  the causal influence of the preparation on the measurement outcome is mediated by the {\em ontic state} of the system, that is, a full specification of the system's physical properties.
We denote the space of ontic states by $\Lambda$.
It is presumed that when the preparation $P$ is implemented, the ontic state of the system, $\lambda\in \Lambda$, is sampled from a probability distribution $\mu(\lambda|P)$, and when the system is subjected to the measurement $M$, the outcome $X$ is distributed as $\xi(X|M,\lambda)$. Finally, for the model to reproduce the experimental statistics, we require that
\begin{equation}
\sum_{\lambda\in \Lambda}\xi(X|M,\lambda)\mu(\lambda|P) = p(X|M,P).
\label{empirical}
\end{equation}

A general discussion of the assumption of noncontextuality is provided in \suppref{A}, but one can understand the concept through the concrete example we consider here (based on a construction from Sec.~V of ref.~\onlinecite{spekkens05}).

Suppose there is a measurement procedure, $M_*$, that is operationally indistinguishable from a fair coin flip: it always gives a uniformly random outcome regardless of the preparation procedure,
\beq
p(X=0,1|M_*,P)=\frac{1}{2},\; \forall P\in\mathcal{P}.
\label{opequiv1}
\eeq
In this case, noncontextuality dictates that in the underlying model, the measurement should also give a uniformly random outcome regardless of the ontic state of the system,
\beq
\xi(X=0,1|M_*,\lambda)=\frac{1}{2},\; \forall \lambda \in\Lambda.
\label{NCimpl1}
\eeq
In other words, because $M_*$ appears operationally to be just like a coin flip, noncontextuality dictates that physically it must be just like a coin flip.

The second application of noncontextuality is essentially a time-reversed version of the first.  Suppose there is  a triple of preparation procedures, $P_1$, $P_2$ and $P_3$, that are operationally indistinguishable from one another: no measurement reveals any information about which of these preparations was implemented,
\beq
\forall M \in \mathcal{M}: p(X|M,P_1)=p(X|M,P_2)=p(X|M,P_3).
\label{opequiv2}
\eeq
In this case, noncontextuality dictates that in the underlying model, the ontic state of the system does not contain any information about which of these preparation procedures was implemented,
\beq
\forall \lambda \in \Lambda: \mu(\lambda|P_1)=\mu(\lambda|P_2)=\mu(\lambda|P_3).
\label{NCimpl2}
\eeq
In other words, because it is impossible, operationally, to extract such information, noncontextuality dictates that physically, the information is not present in the system.

Suppose that $M_*$ can be realized as a uniform mixture of three other binary-outcome measurements, denoted $M_{1}$, $M_{2}$ and $M_{3}$.  That is, one implements $M_*$ by uniformly sampling $t\in\{1,2,3\}$, implementing $M_t$, then outputting its outcome as the outcome of $M_*$.
Finally, suppose that each preparation $P_t$ can be realized as the equal mixture of two other preparation procedures, denoted $P_{t,0}$ and $P_{t,1}$.
Consider implementing $M_t$ on $P_{t,b}$, and consider the average \emph{degree of correlation} between the measurement outcome $X$ and the preparation variable $b$:
\beq
A\equiv \frac{1}{6} \sum_{t\in \{ 1,2,3\}} \sum_{b\in\{0,1\}} p(X=b|M_{t},P_{t,b}).
\label{defnA}
\eeq
We now show that noncontextuality implies a nontrivial bound on $A$.

The proof is by contradiction.  In order to have perfect correlation on average, we require perfect correlation in each term,  which implies that for all ontic states $\lambda$  assigned nonzero probability by $P_{t,b}$, the measurement $M_t$ must respond deterministically with the $X=b$ outcome.
 Given that $P_t$ is an equal mixture of $P_{t,0}$ and $P_{t,1}$, it follows that for all ontic states $\lambda$ assigned nonzero probability by $P_t$, the measurement $M_t$ must have a deterministic response.

But Eq.~\eqref{NCimpl2} (which follows from the assumption of noncontextuality) asserts that the preparations $P_1$, $P_2$ and $P_3$ must assign nonzero probability to precisely the {\em same} set of ontic states.
 Therefore, to achieve perfect correlation on average, each measurement must respond deterministically to {\em all} the ontic states in this set.

Now note that by the definition of $M_*$, the probability of its outcome $X=b$ is $\xi(X=b|M_*,\lambda) = \frac{1}{3}\sum_{t\in \{1,2,3\}} \xi(X=b|M_t,\lambda)$.  But then Eq.~\eqref{NCimpl1} (which follows from the assumption of noncontextuality) says
\beq
\frac{1}{3}\sum_{t\in \{1,2,3\}} \xi(X=b|M_t,\lambda) = \frac{1}{2}.
\label{constraint}
\eeq

For each deterministic assignment of values, $(\xi(X=b|M_1,\lambda),\xi(X=b|M_2,\lambda),\xi(X=b|M_3,\lambda))\in \{ (0,0,0),(0,0,1), \dots, (1,1,1)\}$, the constraint of Eq.~\eqref{constraint} is violated. It follows, therefore, that for a given $\lambda$, one of $M_1$, $M_2$ or $M_3$ must fail to have a deterministic response, contradicting the requirement for perfect correlation on average. This concludes the proof.

The precise (i.e. tight) bound is
\beq
A \le \frac{5}{6},
\label{maininequality}
\eeq
as we demonstrate in \suppref{B}. This is our noncontextuality inequality.

\section{Quantum violation of the inequality}
Quantum theory predicts there is a set of preparations and measurements on a qubit 
having the supposed properties and achieving $A=1$, the logical maximum.  Take the $M_t$ to be 
represented by the observables $\vec{\sigma} \cdot \hat{n}_t$ where $\vec{\sigma}$ is the vector of Pauli operators and the unit vectors $\{ \hat{n}_1,\hat{n}_2, \hat{n}_3 \}$ are separated by 120$^{\circ}$ in the $\hat{x}-\hat{z}$ plane of the Bloch sphere of qubit states\cit{nielsenchuang}. The $P_{t,b}$ are  the eigenstates of these observables, where we associate the positive eigenstate $ |\text{+}\hat{n}_t\rangle \langle +\hat{n}_t|$ with $b=0$.
To see that the statistical equivalence of Eq.~\eqref{opequiv1} is satisfied, it suffices to note that
\beq
 \frac{1}{3} |\text{+}\hat{n}_1\rangle \langle\text{+}\hat{n}_1 |+\frac{1}{3} |\text{+}\hat{n}_2\rangle \langle \text{+}\hat{n}_2 |+\frac{1}{3} |\text{+}\hat{n}_3\rangle \langle\text{+}\hat{n}_3 | =\frac{1}{2}\idn,
 \label{qopequiv1}
\eeq
and to recall that for any density operator $\rho$, ${\rm tr}(\rho \frac{1}{2}\idn)=\frac{1}{2}$.
To see that the statistical equivalence of Eq.~\eqref{opequiv2} is satisfied, it suffices to note that for all pairs $t,t'\in\{1,2,3\}$,
\beqa
&  \frac{1}{2} |\text{+}\hat{n}_t\rangle \langle\text{+}\hat{n}_t | +\frac{1}{2} |{-} \hat{n}_t\rangle \langle{-}\hat{n}_t |\nonumber\\
&=\frac{1}{2} |\text{+}\hat{n}_{t'}\rangle \langle\text{+}\hat{n}_{t'} | +\frac{1}{2} |{-}\hat{n}_{t'}\rangle \langle{-}\hat{n}_{t'} | \label{qopequiv2},
\eeqa
which asserts that the average density operator for each value of $t$ is the same, and therefore leads to precisely the same statistics for all measurements.
Finally, it is clear that the outcome of the measurement of $\vec{\sigma} \cdot \hat{n}_t$ is necessarily perfectly correlated with whether the state was $ |\text{+}\hat{n}_t\rangle \langle \text{+}\hat{n}_t |$ or  $|{-}\hat{n}_t\rangle \langle {-}\hat{n}_t |$, so that $A=1$.

These quantum measurements and preparations are what we seek to implement experimentally, so we refer to them as {\em ideal}, and denote them by $M^{\rm i}_t$ and $P^{\rm i}_{t,b}$.

Note that our noncontextuality inequality can accommodate noise in both the measurements and the preparations, up to the point where the average of 
$p(X=b|M_t,P_{t,b})$ drops below $\frac{5}{6}$.  It is in this sense that our inequality does not presume the idealization of noiseless measurements.

%% file: pseudotheory.tex
\section{Contending with the lack of exact operational equivalence}
The actual preparations and measurements in the experiment, which we call the {\em primary} procedures and denote by $P^{\rm p}_{1,0}$, $P^{\rm p}_{1,1}$, $P^{\rm p}_{2,0}$, $P^{\rm p}_{2,1}$, $P^{\rm p}_{3,0}$, $P^{\rm p}_{3,1}$ and $M^{\rm p}_1$, $M^{\rm p}_2$, $M^{\rm p}_3$, necessarily deviate from the ideal versions and consequently their mixtures, that is, $P^{\rm p}_1$, $P^{\rm p}_2$, $P^{\rm p}_3$ and $M^{\rm p}_*$, fail to achieve strict equality in Eqs.~\eqref{opequiv1} and \eqref{opequiv2}.

We solve this problem as follows.
From the outcome probabilities on the six primary preparations,
one can infer the outcome probabilities on the entire family of probabilistic mixtures of these.
It is possible to find within this family many sets of six preparations, $P^{\rm s}_{1,0}$, $P^{\rm s}_{1,1}$, $P^{\rm s}_{2,0}$, $P^{\rm s}_{2,1}$, $P^{\rm s}_{3,0}$, $P^{\rm s}_{3,1}$, which define mixed preparations $P^{\rm s}_1$, $P^{\rm s}_2$, $P^{\rm s}_3$ that satisfy the operational equivalences of Eq.~\eqref{opequiv2}  \emph{exactly}.
We call the $P^{\rm s}_{t,b}$ {\em secondary} preparations.
We can define secondary measurements $M^{\rm s}_1$, $M^{\rm s}_2$, $M^{\rm s}_3$ and their uniform mixture $M^{\rm s}_{*}$ in a similar fashion.
 The essence of our approach, then, is to identify such secondary sets of procedures and use {\em these} to calculate $A$.  If quantum theory is correct, then we expect to get a value of $A$ close to 1 if and only if we can find suitable secondary procedures that are close to the ideal versions.

To test the hypothesis of noncontextuality, one must allow for the possibility that the experimental procedures {\em do not} admit of  a quantum model.  Nonetheless, for pedagogical purposes, we will first provide the details of how one would construct the secondary sets under the assumption that all the experimental procedures do admit of a quantum model.

\begin{figure}
  \centering
  \includegraphics[width=0.48\textwidth]{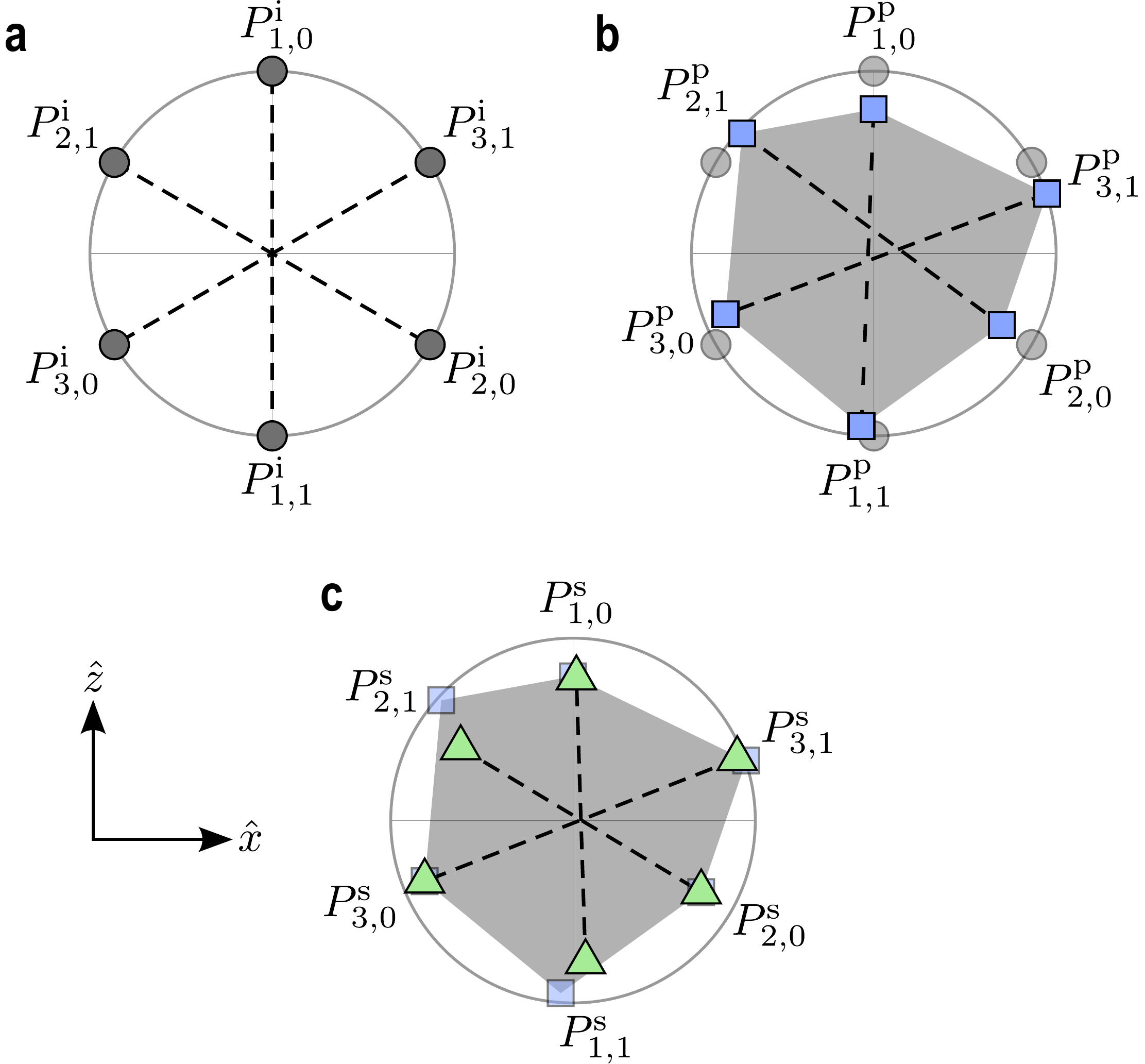}
  \caption{Illustration of our solution to the problem of the failure to achieve strict operational equivalences of preparations (under the simplifying assumption that these are confined to the $\hat{x}-\hat{z}$ plane of the Bloch sphere).  For a given pair, $P_{t,0}$ and $P_{t,1}$, the midpoint along the line connecting the corresponding points represents their equal mixture, $P_t$.  \textbf{a}, The target preparations $P^{\rm i}_{t,b}$, with the coincidence of the midpoints of the three lines illustrating that they satisfy the operational equivalence \eqref{opequiv2} exactly. \textbf{b}, Illustration of how errors in the experiment (exaggerated in magnitude) will imply that the realized preparations $P^{\rm p}_{t,b}$ (termed primary) will deviate from the ideal. The lines indicate that not only do these preparations fail to satify the operational equivalence \eqref{opequiv2}, but since the lines do not meet, no mixtures of the $P^{\rm p}_{t,0}$ and $P^{\rm p}_{t,1}$ can be found at a single point independent of $t$.  The set of preparations corresponding to probabilistic mixtures of the $P^{\rm p}_{t,b}$ are depicted by the grey region.
  \textbf{c}, Secondary preparations $P^{\rm s}_{t,b}$ have been chosen from this grey region, with the coincidence of the midpoints of the three lines indicating that the operational equivalence \eqref{opequiv2} has been restored.  Note that we require only that the mixtures of the three pairs of preparations be the same, not that they correspond to the completely mixed state.}
  \label{fg:convex_combs}
\end{figure}

In Fig.~1, we describe the construction of secondary preparations in a simplified example of six density operators that deviate from the ideal states only {\em within} the $\hat{x}-\hat{z}$ plane of the Bloch sphere.

In practice, the six density operators realized in the experiment will not quite lie in a plane.
We use the same idea to contend with this, but with one refinement:
we supplement our set of ideal preparations with two additional ones, denoted $P^{\rm i}_{4,0}$ and $P^{\rm i}_{4,1}$  corresponding to
the two eigenstates of $\vec{\sigma}\cdot \hat{y}$.  The two procedures that are actually realized in the experiment are denoted $P^{\rm p}_{4,0}$ and $P^{\rm p}_{4,1}$ and are considered supplements to the primary set.  We then search for our six secondary preparations among the probabilistic mixtures of this supplemented set of primaries rather than among the probabilistic mixtures of the original set.  Without this refinement, it can happen that one cannot find six secondary preparations that are close to the ideal versions, as we explain in \suppref{C}.

The scheme for defining secondary measurement procedures is also described in \suppref{C}.  Analogously to the case of preparations, one contends with deviations from the plane by supplementing the ideal set with the observable $\vec{\sigma}\cdot \hat{y}$.

Note that in order to identify which density operators have been realized in an experiment,
 the set of measurements must be complete for state tomography\cit{james01}.  Similarly, to identify which sets of effects
 have been realized,
 the set of preparations must be complete for measurement tomography\cit{lundeen09}.  However, the original ideal sets fail to be tomographically complete because they are restricted to a plane of the Bloch sphere, and an effective way to complete them is to add the observable $\vec{\sigma}\cdot \hat{y}$ to the measurements and its eigenstates to the preparations.
Therefore, even if we did not already need to supplement these ideal sets for the purpose of providing greater leeway in the construction of the secondary procedures, we would be forced to do so in order to ensure that one can achieve tomography.

The relevant procedure here is not quite state tomography in the usual sense, since we want to allow for systematic errors in the measurements as well as the preparations. Hence the task\cit{stark14,stark14b} is to find a set of qubit density operators, $\rho_{t,b}$, and POVMs, $\{E_{X|t}\}$, that together make the measured data as likely as possible (we cannot expect ${\rm tr}(\rho_{t,b}E_{X|t})$ to match the measured relative frequencies exactly due to the finite number of experimental runs).

To analyze our data in a manner that does not prejudice which model---noncontextual, quantum, or otherwise---does justice to it, we must search for representations of the preparations and measurements not amongst density operators and sets of effects, but rather their more abstract counterparts in the formalism of generalised probabilistic theories\cit{hardy01,barrett07}, called generalised states and effects. The assumption that the system is a qubit is replaced by the strictly weaker assumption that three two-outcome measurements are tomographically complete. (In generalised probabilistic theories, a set of measurements are called tomographically complete if their statistics suffice to determine the state.)
We take these states and effects as estimates of our primary preparations and measurements, and we define our estimate of the secondary procedures in terms of these, which in turn are used to calculate our estimate for $A$.
We explain how the raw data is fit to a set of generalised states and effects in \suppref{D}.  We characterize the quality of this fit with a $\chi^2$ test.

%% file: exp.tex
\begin{figure}
  \centering
  \includegraphics[width=0.48\textwidth]{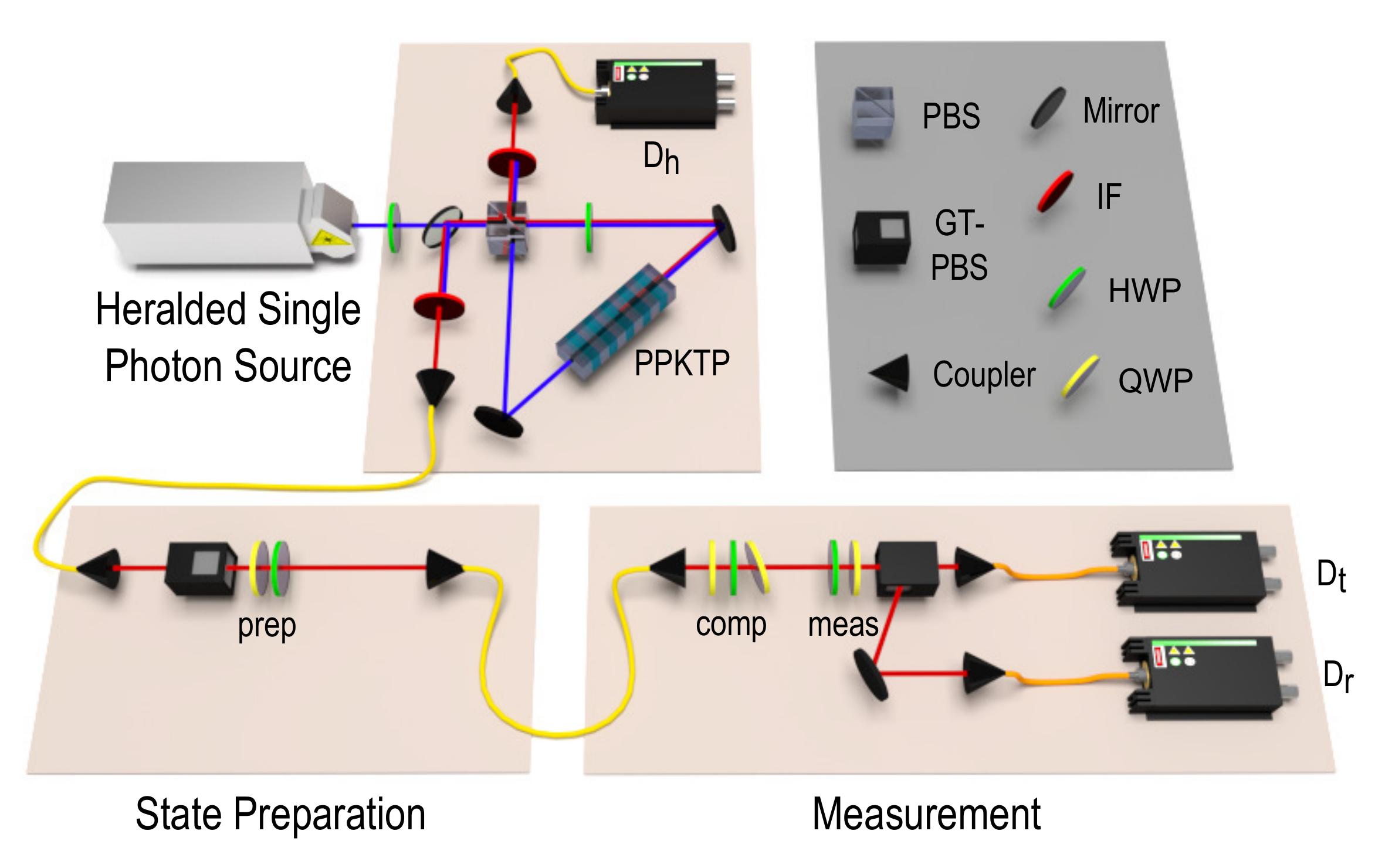}
  \caption{The experimental setup. Polarization-separable photon pairs are created via parametric downconversion, and detection of a photon at $D_h$ heralds the presence of a single photon. The polarization state of this photon is prepared with a polarizer and two waveplates (prep). A single-mode fibre is a spatial filter that decouples beam deflections caused by the state-preparation \text{and measurement} waveplates from the coupling efficiency into the detectors. Three waveplates (comp) are set to undo the polarization rotation caused by the fibre. Two waveplates (meas), a polarizing beamsplitter, and detectors $D_r$ and $D_t$ perform a two-outcome measurement on the state. PPKTP, periodically poled potassium titanyl phosphate; PBS, polarizing beamsplitter; GT-PBS, Glan-Taylor polarizing beamsplitter; IF, interference filter; HWP, half-waveplate; QWP, quarter-waveplate.}
  \label{fg:setup}
\end{figure}
\begin{figure}
  \centering
  \includegraphics[width=0.48\textwidth]{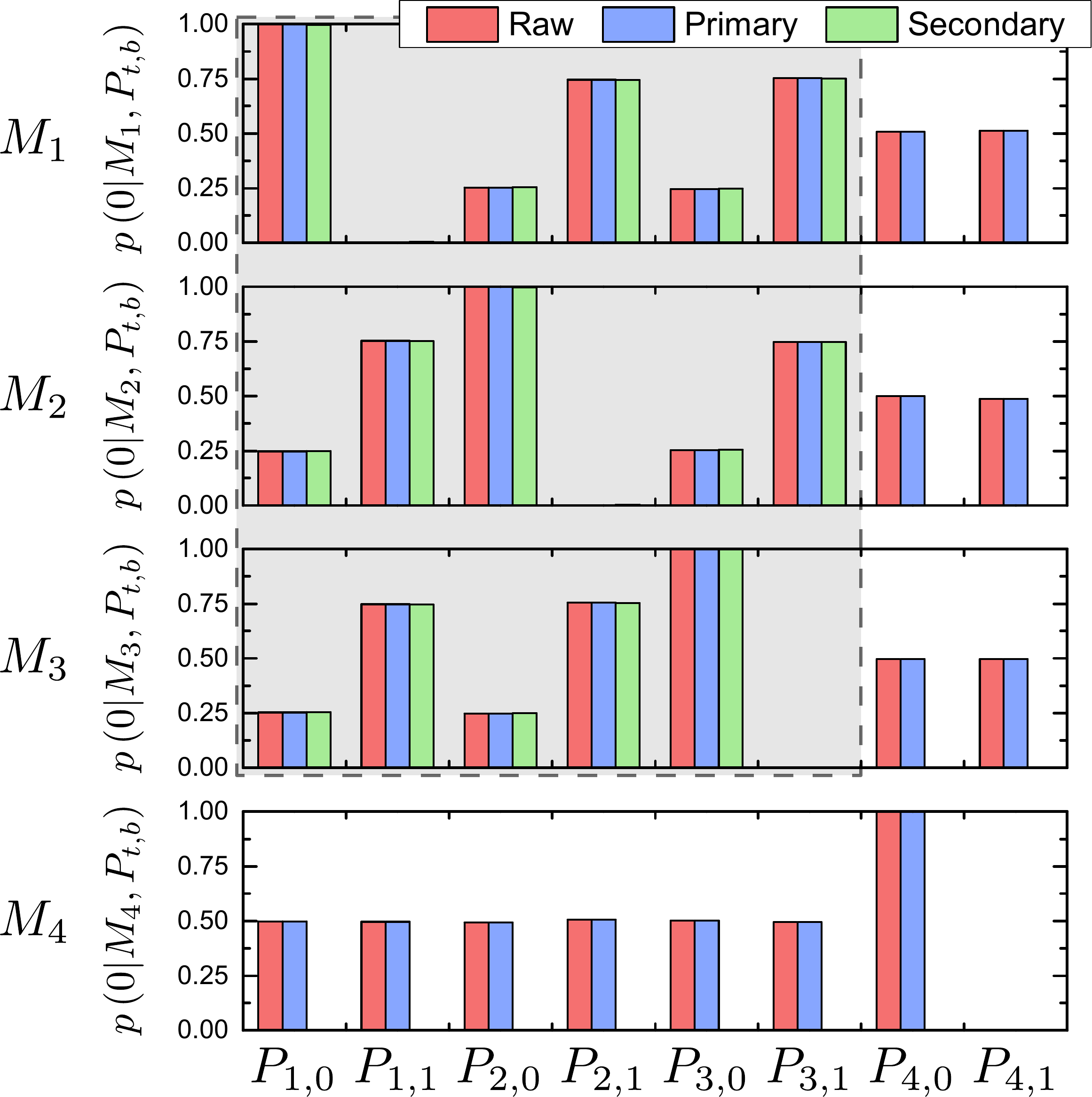}
  \caption{For every measurement-preparation pair, the probability of obtaining outcome 0 in the measurement. Red bars are relative frequencies calculated from the raw counts, blue bars are our estimates of the outcome probabilities of the primary measurements on the primary preparations obtained from a best-fit of the raw data, and green bars are our estimates of the outcome probabilities of the secondary measurements on the secondary preparations. The shaded grey background highlights the measurements and preparations for which secondary procedures were found.
  Error bars are not visible on this scale, neither are discrepancies between the obtained probabilities and the ideal values thereof,
  which are at most $0.013$; statistical error due to Poissonian count statistics is at most $0.002$.}
  \label{fg:freq_data}
\end{figure}
\begin{figure}
  \centering
  \includegraphics[width=0.48\textwidth]{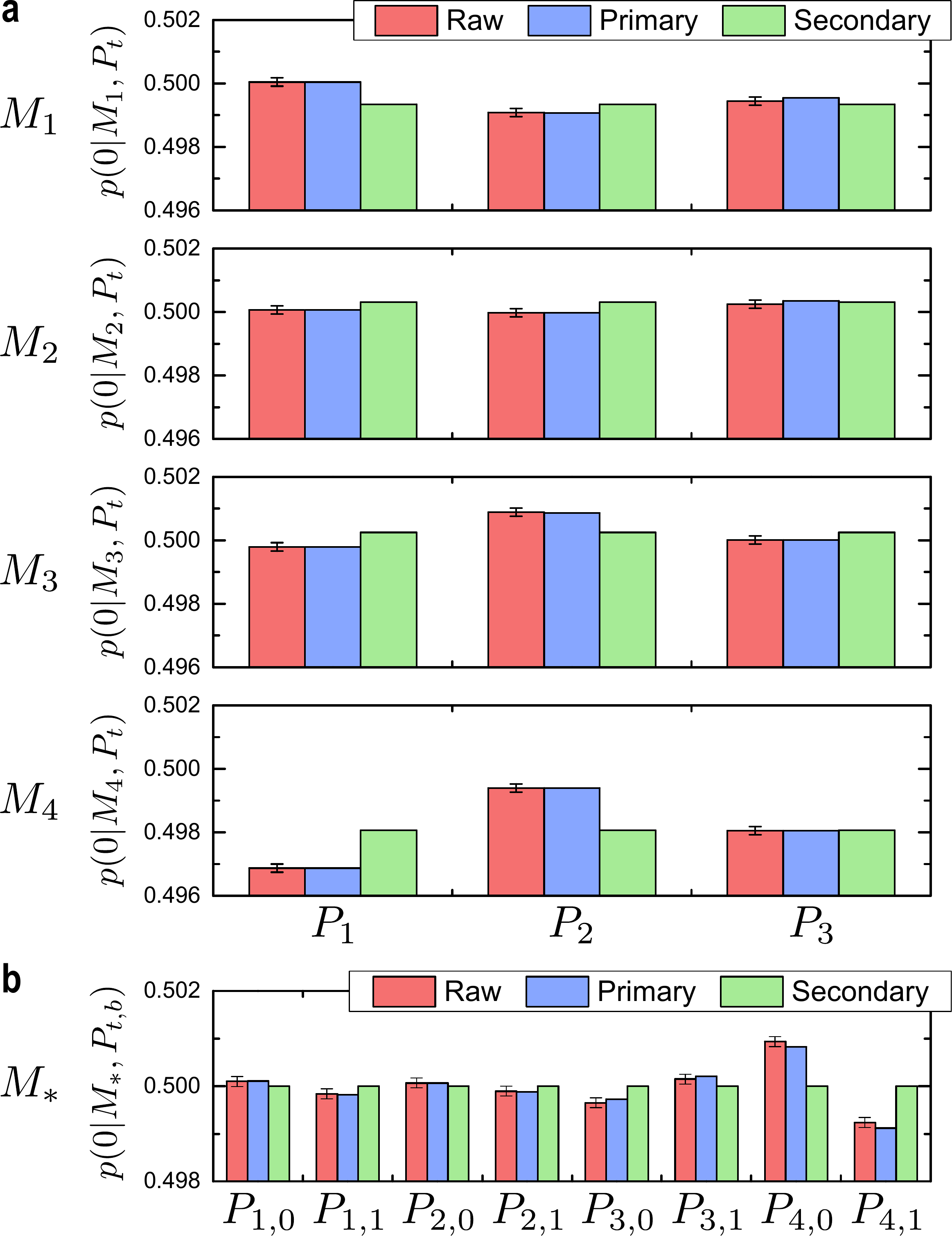}
  \caption{Operational statistics for raw, primary, and secondary preparations and measurements, averaged over 100 experimental runs. \textbf{a}, The probabilities of the primary measurements (blue bars) differ depending on which of the three mixed preparations $P_1^{\rm p}$, $P_2^{\rm p}$, and $P_3^{\rm p}$ are measured. These probabilities are within error of the raw data (red bars), indicating a GPT in which three two-outcome measurements are tomographically complete fits the data well. Probabilities for primary measurements on the secondary preparations (green bars) are independent of the preparation, hence the secondary preparations satisfy Eq.~\eqref{opequiv2}. Note that one expects these probabilities to deviate from 0.5. In the example of Fig.~\ref{fg:convex_combs}\textbf{c}, this corresponds to the fact that the intersection of the lines is not the completely mixed state. \textbf{b}, Outcome probabilities of measurement $M_*$ on the eight preparations. Red bars are raw data, blue bars are the measurement $M_*^{\rm p}$ on the primary preparations, and green bars are $M_*^{\rm s}$ on the primary preparations. Regardless of the input state, $M_*^{\rm s}$ returns outcome 0 with probability 0.5, hence it is operationally indistinguishable from a fair-coin flip (Eq.~\eqref{opequiv1}). Error bars in all plots are calculated assuming Poissonian count statistics.}
  \label{fg:opequivs}
\end{figure}
\begin{figure}
  \centering
  \includegraphics[width=0.48\textwidth]{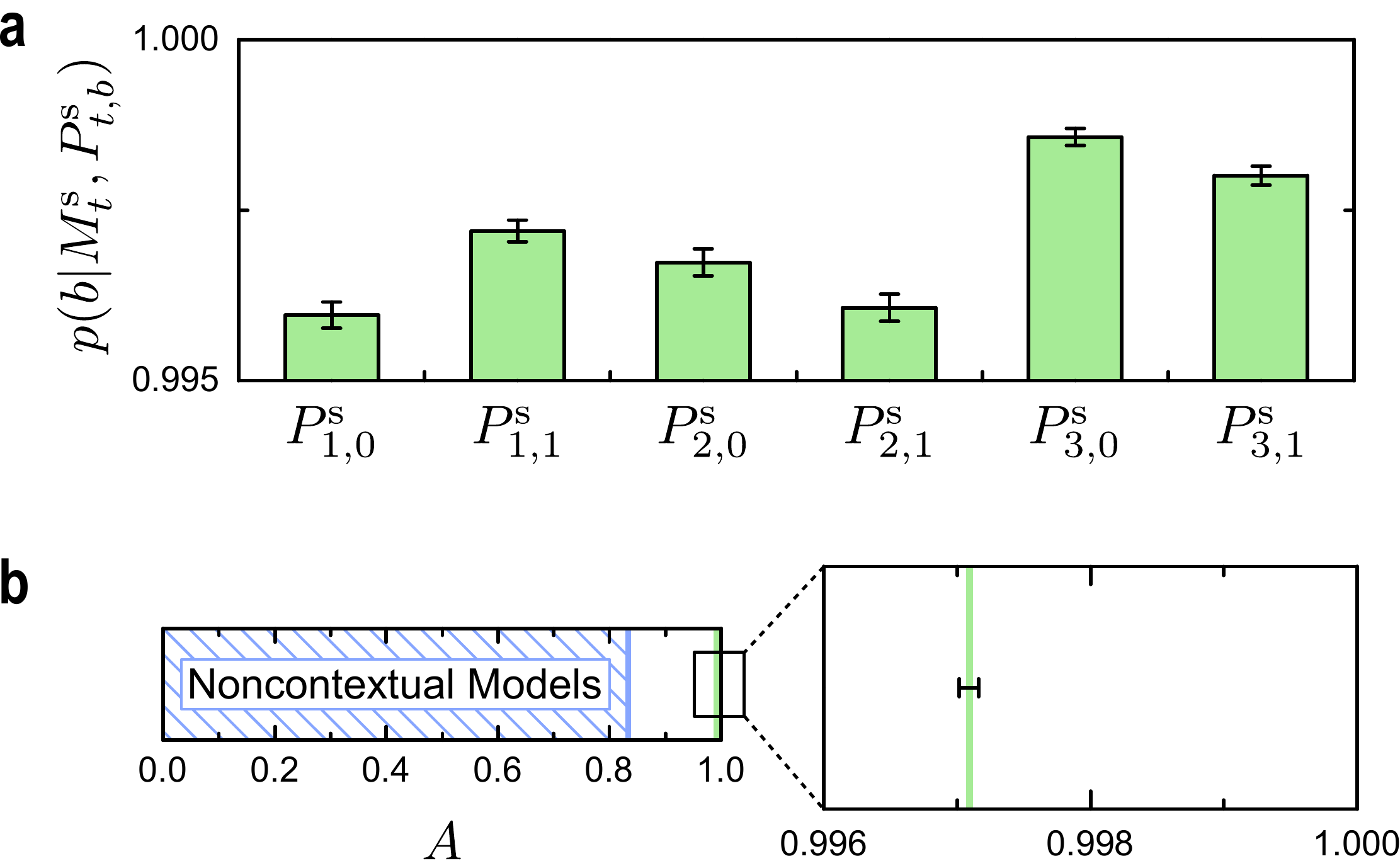}
  \caption{\textbf{a}, Values of the six degrees of correlation in Eq.~\eqref{maininequality}, averaged over 100 experimental runs. \textbf{b}, Average measured value for $A$ contrasted with the noncontextual bound $A=5/6$. We find $A=0.99709\pm0.00007$, which violates the noncontextual bound by 2300$\sigma$. Error bars in both plots represent the standard deviation in the average of the measured values over the 100 experimental runs.}
  \label{fg:data}
\end{figure}

\section{Experiment}
We use the polarization of single photons to test our noncontextuality inequality. The set-up, shown in Fig.~\ref{fg:setup}, consists of a heralded single-photon source\cit{kim06,fedrizzi07,biggerstaff09}, polarization-state preparation and polarization measurement. We generate photons using spontaneous parametric downconversion and prepare eight polarization states using a polarizer followed by a quarter-wave plate (QWP) and half-wave plate (HWP).  The four polarization measurements are performed using a HWP, QWP and polarizing beamsplitter. Photons are counted after the beamsplitter and the counts are taken to be fair samples of the true probabilities for obtaining each outcome for every preparation-measurement pair. Since the orientations of the preparation waveplates lead to small deflections of the beam, some information about the preparation gets encoded spatially, and similarly the measurement waveplates create sensitivity to spatial information; a single-mode fibre deals with both issues. For a single experimental run we implement each preparation-measurement pair for 4s (approximately $10^5$ counts). We performed 100 such runs.

Preparations are represented by vectors of raw data specifying the relative frequencies of outcomes for each measurement, uncertainties on which are calculated assuming Poissonian uncertainty in the photon counts. For each run, the raw data is fit to a set of states and effects in a GPT in which three binary-outcome measurements are tomographically complete.  This is done using a total weighted least-squares method\cit{krystek07,numrec}.  The average $\chi^2$ over the 100 runs is $3.9\pm0.3$, agreeing with the expected value of 4, and indicating that the model fits the data well.
The fit returns a $4 \times 8$ matrix that serves to define the 8 GPT states and 4 GPT effects, which are our estimates of the primary preparations and measurements. The column of this matrix associated to the $t,b$ preparation, which we denote $\mathbf{P}_{t,b}^{\rm p}$, specifies our estimate of the probabilities assigned by the primary preparation $P_{t,b}^{\rm p}$ to outcome `0' of each of the primary measurements.
The raw and primary data are compared in Fig.~\ref{fg:freq_data}. The probabilities are indistinguishable on this scale. We plot the probabilities for $P_1$, $P_2$, and $P_3$ in Fig.~\ref{fg:opequivs}\textbf{a} on a much finer scale. We then see that the primary data are within error of the raw data, as expected given the high quality of the fit to the GPT. However, the operational equivalences of Eqs.~\eqref{opequiv1} and ~\eqref{opequiv2} are not satisfied by our estimates of the primary preparations and measurements, illustrating the need for secondary procedures.

We define the six secondary preparations as probabilistic mixtures of the eight primaries: $\mathbf{P}^{\rm s}_{t,b} = \sum_{t'=1}^4\sum_{b'=0}^1 u_{t',b'}^{t,b} \mathbf{P}_{t',b'}^{\rm p}$, where the $u_{t',b'}^{t,b}$ are the weights in the mixture.
We maximize $C_{\rm P}=\frac{1}{6}\sum_{t=1}^3\sum_{b=0}^1 u_{t,b}^{t,b}$ over valid $u_{t',b'}^{t,b}$  subject to the constraint of
Eq.~\eqref{opequiv2}, that is,
$\frac{1}{2}\sum_b \mathbf{P}^{\rm s}_{1,b}=\frac{1}{2}\sum_b \mathbf{P}^{\rm s}_{2,b}=\frac{1}{2}\sum_b \mathbf{P}^{\rm s}_{3,b}$ (a linear program).
A high value of $C_{\rm P}$ ensures each of the six secondary preparations is close to its corresponding primary. Averaging over 100 runs, we find $C_{\rm P}=0.9969\pm0.0001$, close to the maximum of 1. An analogous linear program to select secondary measurements yields similar results. Fig.~\ref{fg:freq_data} also displays the outcome probabilities for the secondary procedures, confirming that they are close to ideal. Fig.~\ref{fg:opequivs} demonstrates how our construction enforces the operational equivalences.

We analyzed each experimental run separately and found the degree of correlation  $p(X{=}b|M^{\rm s}_t,P^{\rm s}_{t,b})$ for each value of $t$ and $b$.
The averages over the 100 runs are shown in Fig.~\ref{fg:data}\textbf{a} and are all in excess of 0.995.  Averaging over $t$ and $b$ yields
an experimental value $A = 0.99709 \pm 0.00007$, which violates the noncontextual bound of $5/6 \approx 0.833$ by 2300$\sigma$ (Fig.~\ref{fg:data}\textbf{b}).

\section{Discussion}

Using the techniques described here, it is possible to convert proofs of the failure of noncontextuality in quantum theory into experimental tests of noncontextuality that are robust to noise and experimental imprecisions\cit{kunjwal15,pusey15}.  For any phenomenon, therefore, one can determine which of its operational features are genuinely nonclassical.  This is likely to have applications for scientific fields wherein quantum effects are important and for developing novel quantum technologies.

The definition of operational equivalence of preparations (measurements) required them to be statistically equivalent relative to a tomographically complete set of measurements (preparations).
There are two examples of how the assumption of tomographic completeness
is expected {\em not} to hold exactly in our experiment, even if one grants the correctness of quantum theory.

First, our source produces a small multi-photon component.
We measure the $g^{(2)}(0)$ of our source\cit{grangier86} to be $0.0105 \pm 0.0001$ and from this we estimate the ratio of heralded detection events caused by multiple photons to those caused by single photons to be 1:4000.
Regardless of the value of $A$ one presumes for multi-photon events,
one can infer that the value of $A$ we would have achieved had the source been purely single-photon differs from the value given above by at most $10^{-6}$, a difference that does not affect our conclusions.

We also expect the assumption to not hold exactly because of the inevitable coupling of the polarization into the spatial degree of freedom of the photon, which could be caused, for example, by a wedge in a waveplate.
Indeed, we found that if the spatial filter was omitted from the experiment, our fitting routine returned large $\chi^2$ values, which we attributed to the fact that different angles of the waveplates led to different deflections of the beam.

A more abstract worry is that {\em nature} might conflict with the assumption (and prediction of quantum theory) that three independent binary-outcome measurements are tomographically complete for the polarization of a photon.
Our experiment has provided evidence in favour of the assumption insofar as we have fit data from {\em four} measurements to a theory where three are tomographically complete and found a good $\chi^2$ value for the fit.
One can imagine accumulating much more evidence of this sort, but it is difficult to see how any experiment could conclusively vindicate the assumption, given that one can never test all possible measurements. This, therefore, represents the most significant loophole in experimental tests of noncontextuality, and new ideas for how one might seal it or circumvent it represent the new frontier for improving such tests. \color{black}

%% file: supp.tex
\section{Elaboration of the notion of noncontextuality and the idealizations of previous proposals for tests}
\label{suppA}

In this article, we have used the operational notion of noncontextuality proposed in Ref.~\cite{spekkens05}.
According to this notion, one can distinguish noncontextuality for measurements and noncontextuality for preparations.
To provide formal definitions, we must first review the notion of operational equivalence.
Recall that an operational theory specifies a set of physically possible measurements, $\mathcal{M}$, and a set of physically possible preparations, $\mathcal{P}$.  Each measurement $M\in\mathcal{M}$ and preparation $P\in \mathcal{P}$ is assumed to be given as a list of instructions of what to do in the laboratory.  An operational theory also specifies a function $p$, which determines,
 for every preparation $P\in \mathcal{P}$ and every measurement $M\in \mathcal{M}$, the probability distribution over the outcome $X$ of the measurement when it is implemented on that preparation, $p(X|M,P)$.

Two measurement procedures, $M$ and $M'$, are said to be operationally equivalent if they have the same distribution over outcomes for all preparation procedures,
\beqa
p(X|M,P)=p(X|M',P),\;\forall P\in \mathcal{P}
\label{opequivmmts}
\eeqa
Two preparation procedures, $P$ and $P'$, are said to be operationally equivalent if they yield the same distribution over outcomes for all measurement procedures,
\beqa
p(X|M,P')=p(X|M,P),\;\forall M\in \mathcal{M}
\label{opequivpreps}
\eeqa

Any parameters that can be used to describe differences between the measurement procedures in a given operational equivalence class are considered to be part of the {\em measurement context}.   Similarly, parameters that describe differences between preparation procedures in a given operational equivalence class are considered to be part of the {\em preparation context}.  This terminological convention explains the suitability of the term {\em context-independent} or {\em noncontextual} for an ontological model wherein the representation of a given preparation or measurement depends {\em only} on the equivalence class to which it belongs (as defined below).

A {\em tomographically complete set} of preparation procedures, $\mathcal{P}_{\rm tomo} \subseteq \mathcal{P}$,  is defined as one that is sufficient for determining the statistics for any other preparation procedure, and hence is sufficient for deciding operational equivalence of measurements.  In other words, one can equally well define operational equivalence of measurements $M$ and $M'$ by
\beqa
p(X|M,P)=p(X|M',P),\;\forall P\in \mathcal{P}_{\rm tomo}
\label{opequivmmts2}
\eeqa
Similarly, a tomographically complete set of measurement procedures, $\mathcal{M}_{\rm tomo} \subseteq \mathcal{M}$, is defined as one that is sufficient for determining the statistics for any other measurement procedure, and hence is sufficient for deciding operational equivalence of preparations, such that we can define operational equivalence of preparations $P$ and $P'$ by
\beqa
p(X|M,P')=p(X|M,P),\;\forall M\in \mathcal{M}_{\rm tomo}
\label{opequivpreps2}
\eeqa

Note that if the tomographically complete set of preparations for a given system has infinite cardinality, then it is impossible to test operational equivalence experimentally.  In quantum theory, the tomographically complete set for any finite-dimensional system has finite cardinality.

Recall that an ontological model of an operational theory specifies a space $\Lambda$ of ontic states, where an ontic state is defined as a specification of the values of a set of classical variable that mediate the causal influence of the preparation on the measurement.  An ontological model also specifies, for every preparation $P\in \mathcal{P}$, a distribution $\mu(\lambda|P)$.  The idea is that when the preparation $P$ is implemented on a system, it emerges from the preparation device in an ontic state $\lambda$, where $\lambda$ need not be fixed by $P$ but is instead obtained by sampling from the distribution  $\mu(\lambda|P)$.  Similarly, for every measurement $M\in \mathcal{M}$, an ontological model specifies the probabilistic response of the measurement to $\lambda$, specified as a conditional probability $\xi(X|M,\lambda)$ where $X$ is a variable associated to the outcome of $M$.  The idea here is that when an ontic state $\lambda$ is fed into the measurement $M$, it need not fix the outcome $X$, but the outcome is
sampled from the distribution $\xi(X|M,\lambda)$.

The assumption of {\em measurement noncontextuality} is that measurements that are
operationally equivalent should be represented by the same conditional probability distributions in the ontological model,
\beqa
p(X|M,P)=p(X|M',P),\;\forall P\in \mathcal{P}_{\rm tomo} \nonumber\\
\to \xi(X|M,\lambda)=\xi(X|M',\lambda),\; \forall \lambda\in\Lambda.
\label{mnc}
\eeqa
The assumption of {\em preparation noncontextuality} is that preparations that are
operationally equivalent should be represented by the same distributions over ontic states in the ontological model
\beqa
p(X|M,P)=p(X|M,P'),\;\forall M\in \mathcal{M}_{\rm tomo}\nonumber\\
 \to \mu(\lambda|P)=\mu(\lambda|P'),\; \forall \lambda\in\Lambda.
 \label{pnc}
\eeqa
A model is termed simply  {\em noncontextual} if it is measurement noncontextual and preparation noncontextual.

We can summarize this as follows. The grounds for thinking that two measurement procedures are associated with the {\em same} observable, and hence that they are represented equivalently in the noncontextual model, is that they give equivalent statistics for all preparation procedures.  Similarly, two preparations are represented equivalently in the noncontextual model only if they yield the same statistics for all measurements.

The notion of noncontextuality can be understood as a version of Leibniz's Principle of the Identity of Indiscernables, specifically, the {\em physical} identity of {\em operational} indiscernables.  Other instances of the principle's use in physics include the inference from the lack of superluminal signals to the lack of superluminal causal influences (which justifies Bell's assumption of local causality~\cite{bell76}), and Einstein's inference from the operational indistinguishability of accelerating frames and frames fixed in a gravitational field to the physical equivalence of such frames. The question of whether nature admits of a noncontextual model can be understood as whether it adheres to this version of Leibniz's principle, at least within the framework of ontological models that underlies the discussion of noncontextuality.

It is argued in Ref.~\cite{spekkens05} that because the principle
underlying measurement noncontextuality
is the same as the one underlying preparation noncontextuality, if one assumes the first, then one should also assume the second.

As is shown in Ref.~\cite{spekkens05}, the traditional notion of noncontextuality, due to Kochen and Specker~\cite{kochen68}, can be understood as an application of measurement noncontextuality to projective measurements in quantum theory, but involves furthermore an additional assumption that projective measurements should have a {\em deterministic} response to the ontic state.

The idealization of noiseless measurements that we highlighted as a problem of previous attempts to provide an experimental test of noncontextuality can be equivalently characterized as the idealization of deterministic responses of the measurements, as we will now show.

First, recall that determinism is not an assumption of Bell's theorem.  Borrowing an argument from Einstein, Podolsky and Rosen~\cite{Einstein1935}, Bell's 1964 argument~\cite{bell64} leveraged a prediction of quantum theory---that if the same measurement is implemented on two halves of a singlet state, then the outcomes will be perfectly anticorrelated---to {\em derive} the fact that the local outcome-assignments must be deterministic, and from this the first Bell inequality.
But given that experimental correlations are never perfect, no experiment can ever justify determinism. This is why experimentalists use the  Clauser-Horne-Shimony-Holt inequality~\cite{clauser69} to test local causality.

 In Ref.~\cite{spekkens05}, it was shown that if one makes an assumption of noncontextuality for preparations as well as for measurements, then one can also {\em derive} the fact that projective measurements should respond deterministically to the ontic state.
 The inference relies on certain predictions of quantum theory, in particular, that for every projective rank-1 measurement, there is a basis of quantum states that makes its outcome perfectly predictable~\cite{spekkens05,spekkens14}.   However, as in the case of Bell's original inequality, the ideal of perfect predictability is not realized in any experiment.  In particular, perfect predictability only holds under the idealization of noiseless measurements, which is never achieved in practice.

 It is in this sense that previous proposals for testing noncontextuality can be understood as having made an unwarranted idealization of noiseless measurements.

The second idealization that we address in this article concerns the impossibility of realizing any two procedures that satisfy operational equivalence  {\em exactly}.  No two experimental procedures ever give {\em precisely} the same statistics. In formal terms, for any two measurements $M$ and $M'$ that one realizes in the laboratory, it is never the case that one achieves precise equality in Eq.~\eqref{opequivmmts2}. Similarly, for any two preparations $P$ and $P'$ that one realizes in the laboratory, it is never the case that one achieves precisely equality in Eq.~\eqref{opequivpreps2}. In both cases, this is due to the fact that, in practice, one never quite achieves the experimental procedure that one intends to implement.
The problem for an experimental test of noncontextuality, therefore, is that the conditions for applicability of the assumption of noncontextuality (the antecedents in the inferences of Eqs.~\eqref{mnc} and \eqref{pnc})
are, strictly speaking, never satisfied.

\color{black}

\section{Derivation and tightness of the bound in our noncontextuality inequality}
\label{suppB}
\subsection{Derivation of bound}

In the main text, we only provided an argument for why our two applications of the assumption of noncontextuality, Eqs.~(3) and~(5), implied that the quantity $A$ must be bounded away from 1.  Here we show that the explicit value of this bound is $\frac{5}{6}$.

By definition,
\beq
A\equiv \frac{1}{6} \sum_{t\in \{ 1,2,3\}} \sum_{b\in\{0,1\}} p(X=b|M_{t},P_{t,b}).
\eeq
Substituting for $p(X{=}b|M_{t},P_{t,b})$ the expression in terms of the distribution $\mu(\lambda|P_{t,b})$ and the response function $\xi(X=b|M_t,\lambda)$ given in Eq.~(1), we have
\beq
A= \frac{1}{6} \sum_{t\in \{ 1,2,3\}} \sum_{b\in\{0,1\}} \sum_{\lambda\in \Lambda}\xi(X=b|M_t,\lambda)\mu(\lambda|P_{t,b}).
\label{Aexplicit}
\eeq
We now simply note that there is an upper bound on each response function that is independent of the value of $b$, namely,
\beq
\xi(X=b|M_t,\lambda)\le \eta(M_t,\lambda),
\eeq
where
\beq
\eta(M_t,\lambda)\equiv \max_{b'\in \{0,1\}} \xi(X=b'|M_t, \lambda).
\eeq
We therefore have
\beq
A\le \frac{1}{3} \sum_{t\in \{ 1,2,3\}} \sum_{\lambda\in \Lambda}\eta(M_t,\lambda)\left( \frac{1}{2} \sum_{b\in\{0,1\}}  \mu(\lambda|P_{t,b})\right),
\eeq
Recalling that $P_t$ is an equal mixture of $P_{t,0}$ and $P_{t,1}$, so that
\beq
\mu(\lambda|P_{t})= \frac{1}{2} \mu(\lambda|P_{t,0}) + \frac{1}{2} \mu(\lambda|P_{t,1}),
\label{mixmu}
\eeq
we can rewrite the bound as simply
\beq
A\le \frac{1}{3} \sum_{t\in \{ 1,2,3\}}  \sum_{\lambda\in \Lambda} \eta(M_t,\lambda) \mu(\lambda|P_{t}).
\label{Aexplicit2}
\eeq
But recalling Eq.~(5) from the main text,
\beq
\forall \lambda \in \Lambda: \mu(\lambda|P_1)=\mu(\lambda|P_2)=\mu(\lambda|P_3),
\label{NCimpl2supp}
\eeq
we see that the distribution $\mu(\lambda|P_{t})$ is independent of $t$, so we denote it by $\nu(\lambda)$ and rewrite the bound as
\beq
A\le \sum_{\lambda\in \Lambda} \left( \frac{1}{3} \sum_{t\in \{ 1,2,3\}}   \eta(M_t,\lambda) \right)\nu(\lambda).
\eeq
This last step is the first use of noncontextuality in the proof because  Eq.~\eqref{NCimpl2supp} is derived from preparation noncontextuality and the operational equivalence of Eq.~(4).  It then follows that
\beq
A\le \max_{\lambda\in \Lambda} \left( \frac{1}{3} \sum_{t\in \{ 1,2,3\}}   \eta(M_t,\lambda) \right).
\label{Aexplicit3}
\eeq

Therefore, if we can provide a nontrivial upper bound on $\frac{1}{3} \sum_{t}  \eta(M_t,\lambda)$ for an arbitrary ontic state $\lambda$, we obtain a nontrivial upper bound on $A$.
We infer constraints on the possibilities for the triple $(\eta(M_1,\lambda),\eta(M_2,\lambda), \eta(M_3,\lambda))$ from constraints on the possibilities for the triple $(\xi(X\text{=}0|M_1,\lambda),\xi(X\text{=}0|M_2,\lambda), \xi(X\text{=}0|M_3,\lambda))$.

The latter triple
is constrained by Eq.~(7) from the main text, which in the case of $X=0$ reads
\beq
\frac{1}{3}\sum_{t\in \{1,2,3\}} \xi(X\text{=}0|M_t,\lambda) = \frac{1}{2}.
\label{constraintsupp}
\eeq
This is the second use of noncontextuality in our proof, because Eq.~\eqref{constraintsupp} is derived from the operational equivalence of Eq.~(2) and the assumption of measurement noncontextuality.

The fact that the range of each response function is $[0,1]$ implies that the vector $(\xi(X\text{=}0|M_1,\lambda),\allowbreak \xi(X\text{=}0|M_2,\lambda), \xi(X\text{=}0|M_3,\lambda))$ is constrained to the unit cube. The linear constraint \eqref{constraintsupp} implies that these vectors are confined to a two-dimensional plane.  The intersection of the plane and the cube defines the polygon depicted in Fig.~\ref{fcfpolytope}.
\begin{figure}
 \centering
 \includegraphics[scale=0.3]{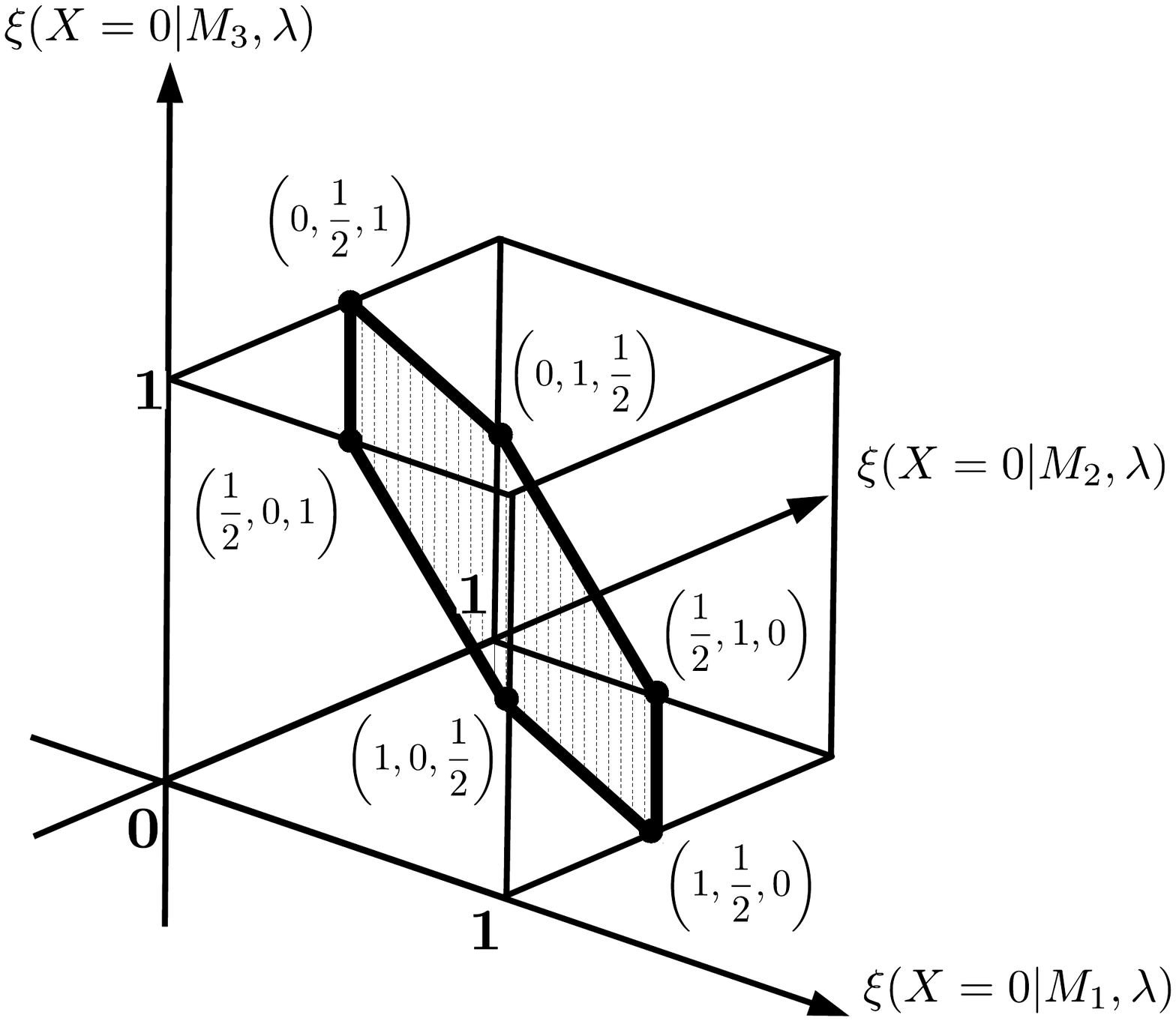}
 \caption{The possible values of $(\xi(X\text{=}0|M_1,\lambda),\xi(X\text{=}0|M_2,\lambda), \xi(X\text{=}0|M_3,\lambda))$.}
\label{fcfpolytope}
\end{figure}
The six vertices of this polygon have coordinates that are a permutation of $(1,\frac{1}{2},0)$.  For every $\lambda$, the vector $(\xi(X\text{=}0|M_1,\lambda),\xi(X\text{=}0|M_2,\lambda), \xi(X\text{=}0|M_3,\lambda))$ corresponds to a point in the convex hull of these extreme points and given that $\frac{1}{3} \sum_{t}  \eta(M_t,\lambda)$  is a convex function of this vector, it suffices to find a bound on the value of this function at the extreme points.  If $\lambda$ is the extreme point $(1,\frac{1}{2},0)$, then we have $(\eta(M_1,\lambda),\eta(M_2,\lambda), \eta(M_3,\lambda))= (1,\frac{1}{2},1)$, and the other extreme points are simply permutations thereof.  It follows that
\beq
\frac{1}{3} \sum_{t}  \eta(M_t,\lambda) \le \frac{5}{6}.
\eeq
Substituting this bound into Eq.~\eqref{Aexplicit3}, we have our result.

\subsection{Tightness of bound: two ontological models}

In this section, we provide an explicit example of a noncontextual ontological model
that saturates our noncontextuality inequality, thus
proving that the noncontextuality inequality is tight, i.e., the upper bound of the
inequality cannot be reduced any further for a noncontextual model.

We also
provide an example of an ontological model that is preparation noncontextual but fails to be measurement noncontextual (i.e. it is measurement {\em contextual}) and that exceeds the bound of our noncontextuality inequality.
This makes it clear that preparation noncontextuality alone does not suffice to justify the precise bound in our inequality, the assumption of measurement noncontextuality is a necessary ingredient as well.   Given that we do not believe preparation noncontextuality on its own to be a reasonable assumption (as discussed in \suppref{A}), we highlight this fact only as a clarification of which features of the experiment are relevant for the particular bound that we obtain.

Note that there is no point inquiring about the bound for models that are measurement noncontextual but preparation contextual because, as shown in Ref.~\cite{spekkens05}, quantum theory admits of models of this type---the ontological model wherein the pure quantum states are the ontic states (the $\psi$-complete ontological model  in the terminology of Ref.~\cite{Harrigan2010}) is of this sort.

For the two ontological models we present, we begin by specifying the ontic state space $\Lambda$.  These are depicted in Figs.~\ref{PNCMNC} and \ref{PNCMC} as pie charts with each slice corresonding to a different element of $\Lambda$.  We specify the six preparations $P_{t,b}$ by the distributions over $\Lambda$ that they correspond to, denoted $\mu(\lambda|P_{t,b})$ (middle left of Figs.~\ref{PNCMNC} and \ref{PNCMC}).  We specify the three measurements $M_t$ by the response functions for the $X=0$ outcome, denoted $\xi(0|M_t, \lambda)$ (top right of Figs.~\ref{PNCMNC} and \ref{PNCMC}).   Finally, we compute the operational probabilities for the various preparation-measurement pairs, using Eq.~(1), and display the results in the $6\times 4$ upper-left-hand corner of Tables~\ref{tablepncmnc} and \ref{tablepncmc}.

\begin{figure}
 \centering
 \includegraphics[scale=0.35]{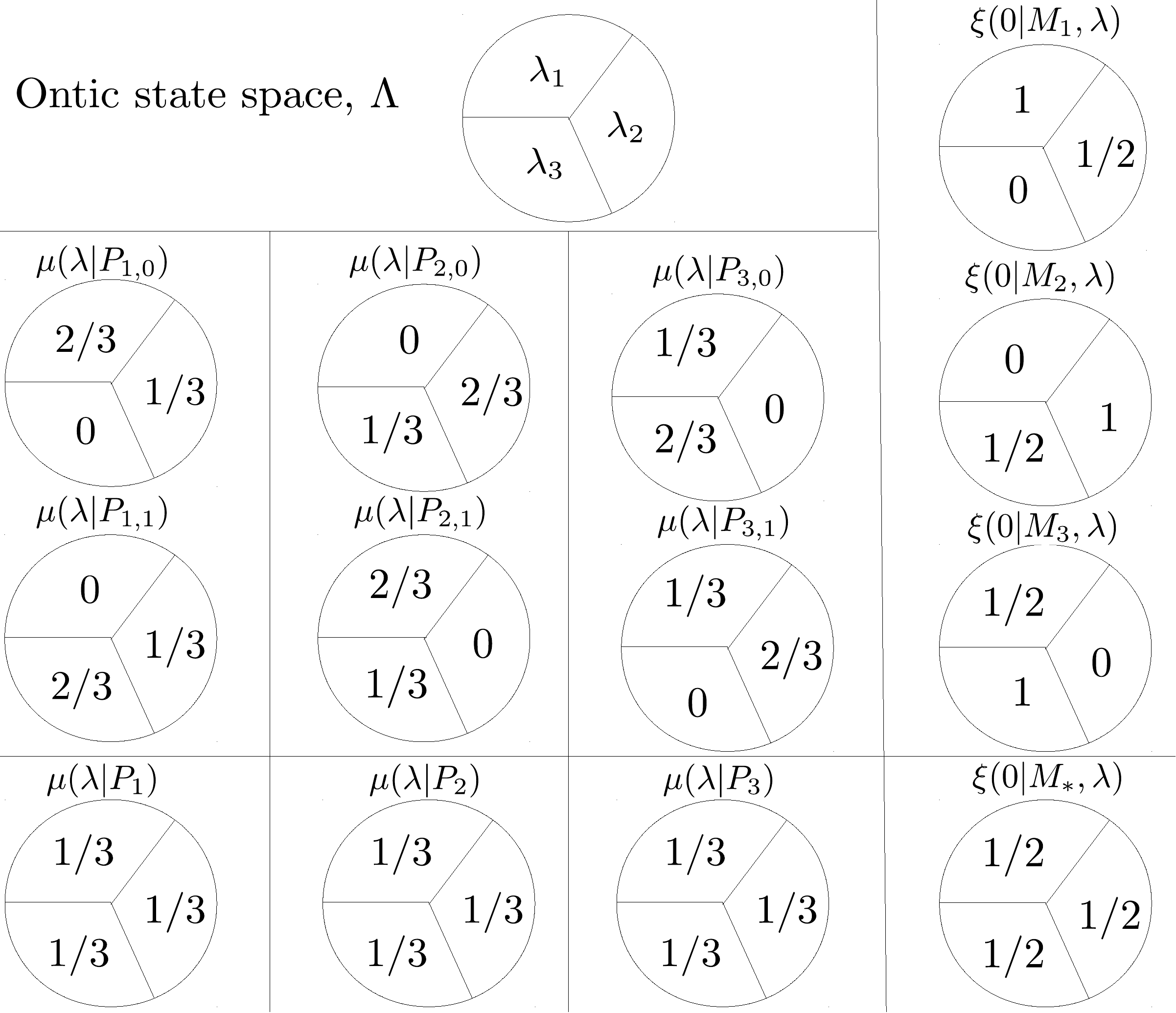}
 \caption{A noncontextual ontological model that saturates the noncontextal bound of our inequality, exhibiting that the bound is tight.}
\label{PNCMNC}
\end{figure}

\begin{table}
\begin{center}
    \begin{tabular}{ | c || c | c | c !{\vrule width 1pt} c | c |}
    \hline
    & $[0|M_1]$ & $[0|M_2]$ & $[0|M_3]$ & $[0|M_*]$ \\ \hline\hline%
    $P_{1,0}$ & \cellcolor{gray} $5/6$ & $1/3$ & $1/3$ & $1/2$ \\ \hline
    $P_{1,1}$ & \cellcolor{gray} $1/6$ & $2/3$ & $2/3$ & $1/2$ \\ \hline

    $P_{2,0}$ & $1/3$ & \cellcolor{gray} $5/6$ & $1/3$ & $1/2$ \\ \hline
    $P_{2,1}$ & $2/3$ & \cellcolor{gray} $1/6$ & $2/3$ & $1/2$ \\ \hline

    $P_{3,0}$ & $1/3$ & $1/3$ & \cellcolor{gray} $5/6$ & $1/2$ \\ \hline
    $P_{3,1}$ & $2/3$ & $2/3$ & \cellcolor{gray} $1/6$ & $1/2$ \\ \hlinewd{1pt}
    $P_1$     & $1/2$ & $1/2$ & $1/2$ & $1/2$ \\ \hline
    $P_2$     & $1/2$ & $1/2$ & $1/2$ & $1/2$ \\ \hline
    $P_3$     & $1/2$ & $1/2$ & $1/2$ & $1/2$ \\ \hline
    \end{tabular}
\end{center}
\caption{Operational statistics from the noncontextual ontological model of Fig.~\ref{PNCMNC}, achieving $A=5/6$. The shaded
cells correspond to the ones relevant for calculating $A$.}
\label{tablepncmnc}
\end{table}

\begin{figure}
 \centering
 \includegraphics[scale=0.35]{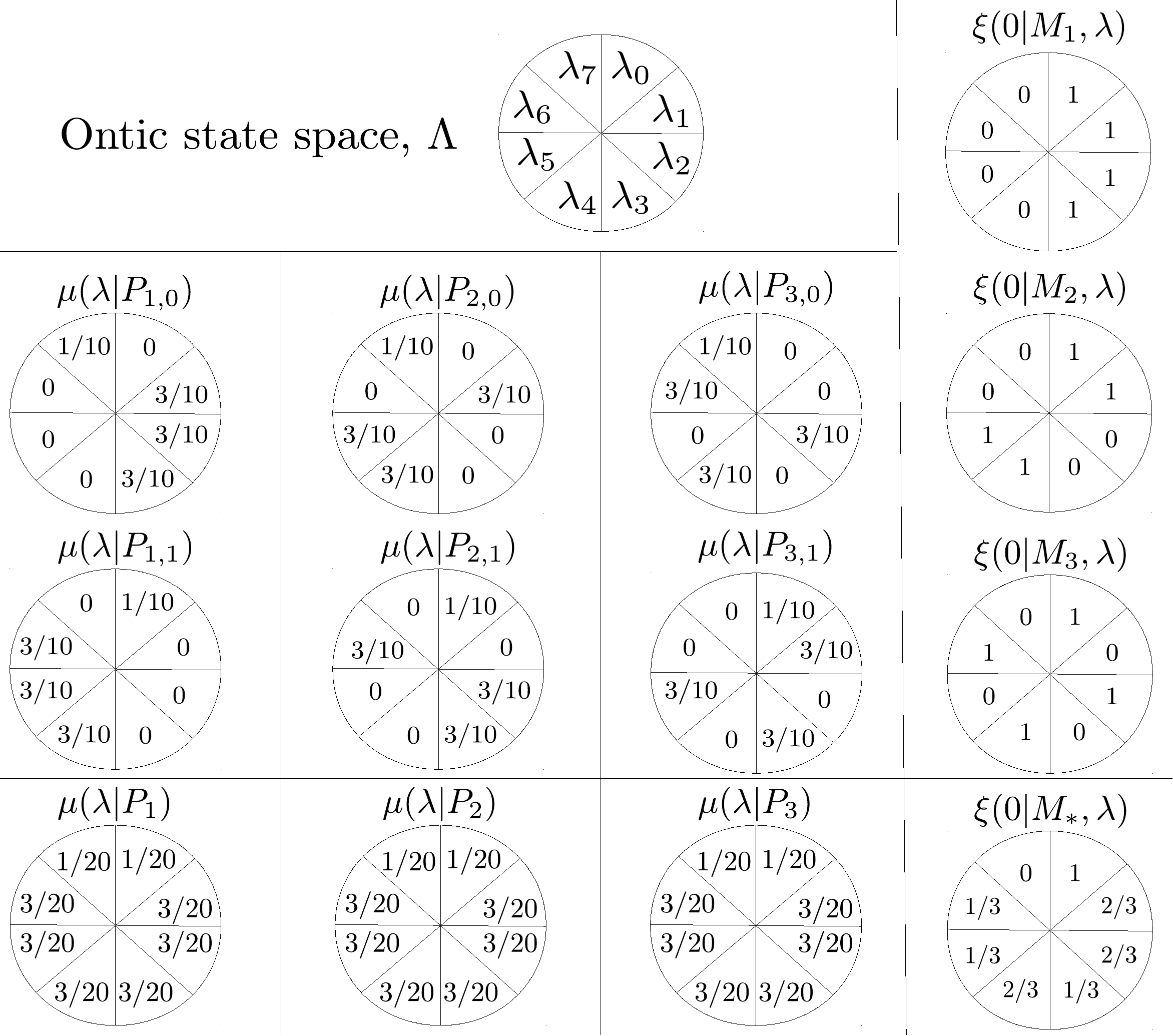}
 \caption{An ontological model that is preparation noncontextual but measurement contextual and that violates our inequality.}
\label{PNCMC}
\end{figure}

\begin{table}
\begin{center}
    \begin{tabular}{ | c || c | c | c !{\vrule width 1pt} c | c |}
    \hline
     & $[0|M_1]$ & $[0|M_2]$ & $[0|M_3]$ & $[0|M_*]$ \\ \hline\hline%
    $P_{1,0}$ & \cellcolor{gray}$9/10$ & $3/10$ & $3/10$ & $1/2$ \\ \hline
    $P_{1,1}$ & \cellcolor{gray}$1/10$ & $7/10$ & $7/10$ & $1/2$ \\ \hline
    $P_{2,0}$ & $3/10$ & \cellcolor{gray}$9/10$ & $3/10$ & $1/2$ \\ \hline
    $P_{2,1}$ & $7/10$ & \cellcolor{gray}$1/10$ & $7/10$ & $1/2$ \\ \hline
    $P_{3,0}$ & $3/10$ & $3/10$ & \cellcolor{gray}$9/10$ & $1/2$ \\ \hline
    $P_{3,1}$ & $7/10$ & $7/10$ & \cellcolor{gray}$1/10$ & $1/2$ \\ \hlinewd{1pt}
    $P_1$     & $1/2$ & $1/2$ & $1/2$ & $1/2$ \\ \hline
    $P_2$     & $1/2$ & $1/2$ & $1/2$ & $1/2$ \\ \hline
    $P_3$     & $1/2$ & $1/2$ & $1/2$ & $1/2$ \\ \hline
    \end{tabular}
\end{center}
\caption{Operational statistics from the preparation noncontextual and measurement contextual ontological model of Fig.~\ref{PNCMC}, achieving $A=9/10$. The shaded cells
correspond to the ones relevant for calculating $A$.}
\label{tablepncmc}
\end{table}

In the remainder of each table, we display the operational probabilities for the effective preparations, $P_t$, which are computed from the operational probabilities for the $P_{t,b}$ and the fact that $P_t$ is the uniform mixture of $P_{t,0}$ and $P_{t,1}$.  We also display the operational probabilities for the effective measurement $M_*$, which is computed from the operational probabilities for the $M_t$ and the fact that $M_*$ is a uniform mixture of $M_1$, $M_2$ and $M_3$.

From the tables, we can verify that our two ontological models imply the operational equivalences that we use in the derivation of our noncontextuality inequality.  Specifically, the three preparations $P_1$, $P_2$ and $P_3$ yield exactly the same statistics for all of the measurements, and the measurement $M_*$ is indistinguishable from a fair coin flip for all the preparations.

Figs.~\ref{PNCMNC} and \ref{PNCMC} also depict $\mu(\lambda|P_t)$ for $t\in\{1,2,3\}$ for each model (bottom left).  These are determined from the $\mu(\lambda|P_{t,b})$ via Eq.~\eqref{mixmu}.  Similarly, the response function $\xi(0|M_*, \lambda)$, which is determined from $\xi(X=b|M_*,\lambda) = \frac{1}{3}\sum_{t\in \{1,2,3\}} \xi(X=b|M_t,\lambda)$, is displayed in each case (bottom right).

Given the operational equivalence of $P_1$, $P_2$ and $P_3$, an ontological model is preparation noncontextual if and only if $\mu(\lambda|P_1)=\mu(\lambda|P_2)=\mu(\lambda|P_3)$ for all $\lambda\in \Lambda$. We see, therefore, that both models are preparation noncontextual.

Similarly given the operational equivalence of $M_*$ and a fair coin flip, an ontological model is measurement noncontextual if and only if $\xi(0|M_*, \lambda)=\frac{1}{2}$ for all $\lambda\in \Lambda$.  We see, therefore, that only the first model is measurement noncontextual.

Note that in the second model, $M_*$ manages to be operationally equivalent to a fair coin flip, despite the fact that when one conditions on a given ontic state $\lambda$, it does not have a uniformly random response.  This is possible only because the set of distributions is restricted in scope, and the overlaps of these distributions with the response functions always generates the uniformly random outcome.  This highlights how an ontological model can do justice to the operational probabilities while failing to be noncontextual.

Finally, using the operational probabilities in the tables, one can compute the value of $A$ for each model.  It is determined entirely by the operational probabilities in the shaded cells. One thereby confirms that $A=\frac{5}{6}$ in the first model, while $A=\frac{9}{10}$ in the second model.

\section{Constructing the secondary procedures from the primary ones}
\label{suppC}

\subsection{Secondary preparations in quantum theory}

As noted in the main text, it is easiest to describe the details of our procedure for defining secondary preparations if we make the assumption that quantum theory correctly describes the experiment.  Further on, we will describe the procedure for a generalised probabilistic theory (GPT).

Fig. 1 in the main text described how to define the secondary preparations if the primary preparations deviate from the ideal only {\em within} the $\hat{x}-\hat{z}$ plane of the Bloch sphere.  Here, we consider the case where the six primary preparations deviate from the ideals within the bulk of the Bloch sphere.
The fact that our proof only requires that the secondary preparations satisfy Eq.~(10) means that the different pairs, $P^{\rm s}_{t,0}$ and $P^{\rm s}_{t,1}$ for $t\in \{1,2,3\}$, need not all mix to the center of the Bloch sphere, but only to the {\em same} state.  It follows that the three pairs need not be coplanar in the Bloch sphere.  Note, however, for any {\em two} values, $t$ and $t'$, the four preparations $P^{\rm s}_{t,0}, P^{\rm s}_{t,1}, P^{\rm s}_{t',0}, P^{\rm s}_{t,1}$ do need to be coplanar.

Any mixing procedure defines a  map from each of the primary preparations $P^{\rm p}_{t,b}$ to the corresponding secondary preparation $P^{\rm s}_{t,b}$, which can be visualized as a motion of the corresponding point within the Bloch sphere.  To ensure that the six secondary preparations approximate well the ideal preparations while also defining mixed preparations $P^{\rm s}_1$, $P^{\rm s}_2$ and $P^{\rm s}_3$ that satisfy the appropriate operational equivalences, the mixing procedure must allow for motion in the $\pm\hat{y}$ direction.  Consider what happens if one tries to achieve such motion {\em without} supplementing the primary set with the eigenstates of $\vec{\sigma}\cdot \hat{y}$.  A given point that is biased towards $-\hat{y}$ can be moved in the $+\hat{y}$ direction by mixing it with another point that has less bias in the $-\hat{y}$ direction.  However, because the primary preparations are widely separated within the $\hat{x}-\hat{z}$ plane, achieving a small motion in $+\hat{y}$ direction in
this fashion comes at the price of a large motion within the $\hat{x}-\hat{z}$ plane, implying a significant motion away from the ideal.  This problem is particularly pronounced if the primary points are very close to coplanar.

The best way to move a given point in the $\pm\hat{y}$ direction is to mix it with a point that is at roughly the same location within the $\hat{x}-\hat{z}$ plane, but displaced in the $\pm\hat{y}$ direction. This scheme, however, would require supplementing the primary set with one or two additional preparations for every one of its elements.    Supplementing the original set with just the two  eigenstates of $\vec{\sigma}\cdot \hat{y}$ constitutes a good compromise between keeping the number of preparations low and ensuring that the secondary preparations are close to the ideal. Because the $\vec{\sigma}\cdot \hat{y}$ eigenstates have the greatest possible distance from the $\hat{x}-\hat{z}$ plane, they can be used to move any point close to that plane in the $\pm \hat{y}$ direction while generating only a modest motion within the $\hat{x}-\hat{z}$ plane.

\subsection{Secondary measurements in quantum theory}

Just as with the case of preparations, we solve the problem of no strict statistical equivalences for measurements by noting that from the primary set of measurements, $M^{\rm p}_1$, $M^{\rm p}_2$ and $M^{\rm p}_3$, one can infer the statistics of a large family of measurements, and one can find three measurements within this family, called the secondary measurements and denoted $M^{\rm s}_1$, $M^{\rm s}_2$ and $M^{\rm s}_3$, such that their mixture, $M^{\rm s}_*$, satisfies the operational equivalence of Eq.~(2) {\em exactly}.  To give the details of our approach, it is again useful to begin with the quantum description.

A geometric visualization of the construction is also possible in this case.
Just as a density operator can be written $\rho = \frac12(\idn + \vec r \cdot \vec\sigma)$ to define a three-dimensional Bloch vector $\vec r$, an effect can be written $E = \frac12(e_0\idn + \vec e\cdot\vec\sigma)$ to define a four-dimensional Bloch-like vector $(e_0, \vec e)$, whose four components we will call the $\hat \idn$, $\hat x$, $\hat y$ and $\hat z$ components.  Note that $e_0={\rm tr}(E)$, while $e_x = {\rm tr}(\vec{\sigma}\cdot \hat{x} E)$ and so forth.  The eigenvalues of $E$ are expressed in terms of these components as $\frac{1}{2}(e_o \pm |\vec{e}|)$.  Consequently, the constraint that $0 \le E \le \idn$ takes the form of three inequalities $0 \le e_o \le 2$, $|\vec{e}| \le e_0$ and $|\vec{e}| \le 2-e_0$.  This corresponds to the intersection of two cones.  For the case $e_y=0$, the Bloch representation of the effect space is three-dimensional and is displayed in Fig.~\ref{Blochcone}.
When portraying binary-outcome measurements associated to a POVM $\{ E, \idn - E\}$ in this representation, it is sufficient to portray the Bloch-like vector $(e_0,\vec{e})$ for outcome $E$ alone, given that the vector for $\idn - E$ is simply $(2-e_0, -\vec{e})$.  Similarly, to describe any mixture of two such POVMs, it is sufficient to describe the mixture of the effects corresponding to the first outcome.
The family of measurements that is defined in terms of the primary set is slightly different than what we had for preparations.  The reason is that each primary measurement on its own generates a family of measurements by probabilistic post-processing of its outcome. If we denote the outcome of the original measurement by $X$ and that of the processed measurement by $X'$, then the probabilistic processing is a conditional probability $p(X'|X)$.  It is sufficient to determine the convexly-extremal post-processings, since all others can be obtained from these by mixing.  For the case of binary outcome measurements considered here, there are just four extremal post-processings: the identity process, $p(X'|X)=\delta_{X',X}$; the process that flips the outcome, $p(X'|X)=\delta_{X',X\oplus 1}$; the process that always generates the outcome $X'=0$, $p(X'|X)=\delta_{X',0}$; and the process that always generates the outcome $X'=1$, $p(X'|X)=\delta_{X',1}$.  Applying these to our three primary measurements, we obtain
eight measurements in all: the two that generate a fixed outcome, the three originals, and the three originals with the outcome flipped.   If the set of primary measurements corresponded to the ideal set, then the eight extremal post-processings would correspond to the observables $0, \idn, \vec{\sigma}\cdot \hat{n_1}, {-}\vec{\sigma}\cdot \hat{n_1}, \vec{\sigma}\cdot \hat{n_2}, {-}\vec{\sigma}\cdot \hat{n_2}, \vec{\sigma}\cdot \hat{n_3}, {-}\vec{\sigma}\cdot \hat{n_3}$.  In practice, the last six measurements will be unsharp.  These eight measurements can then be mixed probabilistically to define the family of measurements from which the secondary measurements must be chosen.  We refer to this family as the {\em convex hull of the post-processings} of the primary set.

We will again start with a simplified example, wherein the primary measurements have Bloch-like vectors with vanishing component along $\hat{y}$, $e_y = 0$, and unit component along $\idn$, $e_0=1$, so that $E = \frac12(\idn + e_x \vec\sigma\cdot \hat{x} + e_z \vec\sigma \cdot \hat{z})$.  In this case, the constraint $0\le E\le \idn$ reduces to $|\vec{e}|\le 1$, which is the same constraint that applies to density operators confined to the $\hat{x}-\hat{z}$ plane of the Bloch sphere.  Here, the only deviation from the ideal is within this plane,
and the construction is precisely analogous to what is depicted in Fig.~1 of the main text.

Unlike the case of preparations, however, the primary measurements can deviate from the ideal in the $\hat \idn$ direction, that is, $E$ may have a component along $\idn$ that deviates from $1$, which corresponds to introducing a state-independent bias on the outcome of the measurement.  This is where the extremal post-processings yielding the constant-outcome measurements corresponding to the observables 0 and $\idn$ come in.  They allow one to move in the $\pm\hat \idn$ direction.

Fig.~\ref{Blochcone} presents an example wherein the primary measurements have Bloch-like vectors that deviate from the ideal not only within the $\hat{x}-\hat{z}$ plane, but in the $\hat{\idn}$ direction as well (it is still presumed, however, that all components in the $\hat{y}$ direction are vanishing).

\begin{figure}
 \centering
 \includegraphics[width=0.48\textwidth]{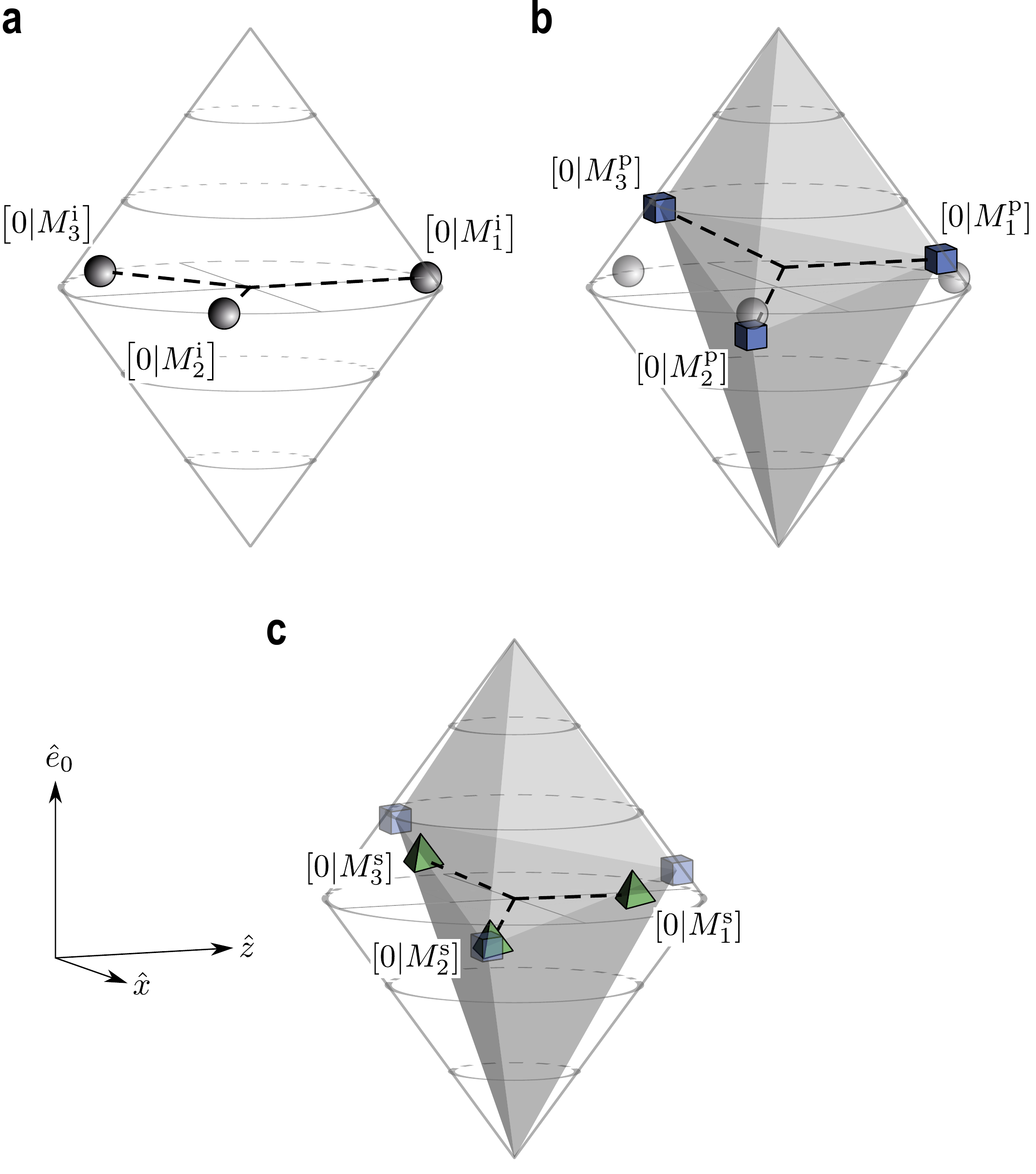}
 \caption{A depiction of the construction of secondary measurements from primary ones in the simplified case where the component along $\hat{y}$ is zero.  For each measurement, we specify the point corresponding to the Bloch representation of its first outcome.  These are labelled $[0|M_1]$, $[0|M_2]$ and $[0|M_3]$.  The equal mixture of these three, labelled $[0|M_*]$, is the centroid of these three points, i.e. the point equidistant from all three. \textbf{a}, The ideal measurements $[0|M^{\rm i}_{t}]$ with centroid at $\idn/2$, illustrating that the operational equivalence~(2) is satisfied exactly. \textbf{b}, Errors in the experiment (exaggerated) will imply that the realized measurements $[0|M^{\rm p}_{t,}]$ (termed primary) will deviate from the ideal, and their centroid deviates from $\idn/2$.  The family of points corresponding to probabilistic mixtures of the $[0|M^{\rm p}_{t}]$ and the
 observables $0$ and $\idn$
  are depicted by the grey region.  (For clarity, we have not depicted the outcome-flipped versions of the three primary measurements, and have not included them in the probabilistic mixtures. As we note in the text, such a restriction still allows for a good construction.)
\textbf{c}, The secondary measurements $M^{\rm s}_{t}$ that have been chosen from this grey region.  They are chosen such that their centroid is at $\idn/2$, restoring the operational equivalence~(2).}
\label{Blochcone}
\end{figure}

In practice, of course, the $\hat{y}$ component of our measurements never vanishes precisely either.  We therefore apply the same trick as we did for the preparations.
We supplement the set of primary measurements with an additional measurement, denoted $M^{\rm p}_4$, that ideally corresponds to the observable $\vec{\sigma}\cdot \hat{y}$.  The post-processing which flips the outcome then corresponds to the observable $-\vec{\sigma}\cdot \hat{y}$.  Mixing the primary measurements with $M^{\rm p}_4$ and its outcome-flipped counterpart
allows motion in the $\pm\hat{y}$ direction within the Bloch cone.

Note that the capacity to move in both the $\hat{y}$ and the $-\hat{y}$ direction is critical for achieving the operational equivalence of Eq.~(2), because if the secondary measurements had a common bias in the $\hat{y}$ direction, they could not mix to the POVM $\{ \idn/2, \idn/2\}$ as Eq.~(9) requires.  For the preparations, by contrast, supplementing the primary set by just {\em one} of the eigenstates of $\vec{\sigma}\cdot \hat{y}$ would still work, given that the mixed preparations $P^{\rm s}_{t}$ do not need to coincide with the completely mixed state $\idn/2$.

The secondary measurements $M^{\rm s}_1$, $M^{\rm s}_2$ and $M^{\rm s}_3$ are then chosen from the convex hull of the post-processings of the $ M^{\rm p}_1,M^{\rm p}_2,M^{\rm p}_3,M^{\rm p}_4$.   Without this supplementation, it may be impossible to find secondary measurements that define an $M^{\rm s}_*$ that satisfies the operational equivalences while providing a good approximation to the ideal measurements.

In all, under the extremal post-processings of the supplemented set of primary measurements, we obtain ten points which ideally correspond to the observables $0, \idn, \vec{\sigma}\cdot \hat{n_1}, {-}\vec{\sigma}\cdot \hat{n_1}, \vec{\sigma}\cdot \hat{n_2}, {-}\vec{\sigma}\cdot \hat{n_2}, \vec{\sigma}\cdot \hat{n_3}, {-}\vec{\sigma}\cdot \hat{n_3}, \vec{\sigma}\cdot \hat{y}$, and  ${-}\vec{\sigma}\cdot \hat{y}$.

Note that the outcome-flipped versions of the three primary measurements are not critical for defining a good set of secondary measurements, and indeed we find that we can dispense with them and still obtain good results.  This is illustrated in the example of Fig.~\ref{Blochcone}.

\subsection{Secondary preparations and measurements in generalised probabilistic theories}

We do not want to presuppose that our experiment is well fit by a quantum description. Therefore instead of working with density operators and POVMs, we work with GPT states and effects, which are inferred from  the matrix $D^{\rm p}$
\begin{equation}
D^{\rm p} =
 \begin{pmatrix}
  p^{1}_{1,0}  & p^{1}_{1,1} & \cdots & p^{1}_{4,0}  & p^{1}_{4,1} \\
  p^{2}_{1,0}  & p^{2}_{1,1} & \cdots & p^{2}_{4,0}  & p^{2}_{4,1} \\
  p^{3}_{1,0}  & p^{3}_{1,1} & \cdots & p^{3}_{4,0}  & p^{3}_{4,1} \\
  p^{4}_{1,0}  & p^{4}_{1,1} & \cdots & p^{4}_{4,0}  & p^{4}_{4,1}
 \end{pmatrix}.
 \label{prim_matrix}
\end{equation}
where
\beq
p^{t'}_{t,b}\equiv p(0|M^{\rm p}_{t'},P^{\rm p}_{t,b})
\eeq
is the probability of obtaining outcome $0$ in the $t'$th measurement that was actually realized in the experiment (recall that we term this measurement primary and denote it by $M^{\rm p}_{t'}$), when it follows the $(t,b)$th preparation that was actually realized in the experiment (recall that we term this preparation primary and denote it by $P^{\rm p}_{t,b}$).
These probabilities are estimated by fitting the raw  experimental data (which are merely finite samples of the true probabilities) to a GPT; we postpone the description of this procedure to Sec.~\ref{fittingrawdata}.

The rows of the $D^{\rm p}$ matrix define the GPT effects. We denote the vector defined by the $t$th row,  which is associated to the measurement event $[0|M_t^{\rm p}]$ (obtaining the 0 outcome in the primary measurement $M_t^{\rm p}$), by $\mathbf{M}_t^{\rm p}$.
Similarly, the columns of this matrix define the GPT states.  We denote the vector associated to the $(t,b)$th column, which is associated to the primary preparation $P^{\rm p}_{t,b}$, by $\mathbf{P}^{\rm p}_{t,b}$.

As described in the main text, we define the {\em secondary} preparation $P^{\rm s}_{t,b}$ by a probabilistic mixture of the primary preparations.  Thus, the GPT state of the secondary preparation is a vector $\mathbf{P}^{\rm s}_{t,b}$ that is a probabilistic mixture of the $\mathbf{P}^{\rm p}_{t,b}$,
\beq
\mathbf{P}^{\rm s}_{t,b} = \sum_{t'=1}^4\sum_{b'=0}^1 u_{t',b'}^{t,b} \mathbf{P}_{t',b'}^{\rm p},
\eeq
where the $u_{t',b'}^{t,b}$ are the weights in the mixture.

A secondary measurement $M_{t'}^{\rm s}$ is obtained from the primary measurements in a similar fashion, but in addition to probabilistic mixtures, one must allow certain post-processings of the measurements,  in analogy to the quantum case described above.
The set of all post-processings of the primary outcome-0 measurement events has extremal elements consisting of the outcome-0 measurement events themselves together with:
the measurement event that {\em always} occurs (i.e. obtaining outcome `0' or `1'), which is represented by the vector of probabilities where every entry is 1, denoted $\mathbf{1}$; the measurement event that {\em never} occurs (i.e. obtaining neither outcome `0' nor outcome `1'), which is represented by the vector of probabilities where every entry is 0, denoted $\mathbf{0}$; and the outcome-1 measurement events, $[1|M_t^{\rm p}]$,
which is represented by the vector $\mathbf{1}-\mathbf{M}_t^{\rm p}$.

We can therefore define our three secondary outcome-0 measurement events as probabilistic mixtures of the four primary ones as well as the extremal post-processings mentioned above, that is
\beq
\mathbf{M}_t^{\rm s} = \sum_{t'=1}^{4} v_{t'}^t\mathbf{M}_{t'}^{\rm p} + v_{\mathbf{0}}^t \mathbf{0} + v_{\mathbf{1}}^t\mathbf{1} + \sum_{t''=1}^{4} v_{\lnot t''}^t(\mathbf{1}-\mathbf{M}_{t''}^{\rm p}),
\eeq
where for each $t$, the vector of weights in the mixture is $(v_1^t,v_2^t,v_3^t,v_4^t,v_{\mathbf{0}}^t,v_{\mathbf{1}}^t,v_{\lnot 1}^t,v_{\lnot 2}^t,v_{\lnot 3}^t,v_{\lnot 4}^t)$. We see that this is a particular type of linear transformation on the rows.

Again, as mentioned in the discussion of the quantum case, we can in fact limit the post-processing to exclude the
outcome-1 measurement events for
$M_1$, $M_2$ and $M_3$, keeping only the outcome-1 event for $M_4$, and still obtain good results.  Thus we found it sufficient to search for secondary outcome-0 measurement events among those of the form
\beq
\mathbf{M}_t^{\rm s} = \sum_{t'=1}^{4} v_{t'}^t\mathbf{M}_{t'}^{\rm p} + v_{\mathbf{0}}^t \mathbf{0} + v_{\mathbf{1}}^t\mathbf{1} + v_{\lnot 4}^t(\mathbf{1}-\mathbf{M}_{4}^{\rm p}),
\eeq
where for each $t$, the vector of weights in the mixture is $(v_1^t,v_2^t,v_3^t,v_4^t,v_{\mathbf{0}}^t,v_{\mathbf{1}}^t,v_{\lnot 4}^t)$.

\color{black}
Returning to the preparations, we choose the weights $u^{t,b}_{t',b'}$  to maximize the function
\beq
C_{\rm P} \equiv \frac{1}{6}\sum_{t=1}^3\sum_{b=0}^1 u_{t,b}^{t,b}
\eeq
subject to the linear constraint
\beq
\frac{1}{2}\sum_b \mathbf{P}^{\rm s}_{1,b}=\frac{1}{2}\sum_b \mathbf{P}^{\rm s}_{2,b}=\frac{1}{2}\sum_b \mathbf{P}^{\rm s}_{3,b},
\eeq
as noted in the main text.  This optimization ensures that the secondary preparations are as close as possible to the primary ones while ensuring that they satisfy the relevant operational equivalence {\em exactly}.  Table~\ref{tb:prep_weights} reports the weights $u_
{t',b'}^{t,b}$ that were obtained from this optimization procedure, averaged over the 100 runs of the experiment.  As noted in the main text, these weights yield $C_{\rm P}=0.9969\pm0.0001$, indicating that the secondary preparations are indeed very close to the primary ones.

The scheme for finding the weights $(v_1^t,v_2^t,v_3^t,v_4^t,v_{\mathbf{0}}^t,v_{\mathbf{1}}^t,v_{\lnot 4}^t)$ that define the secondary measurements is analogous.
Using a linear program, we find the vector of such weights that maximizes the function
\beq
C_{\rm M}\equiv \frac{1}{3}\sum_{t=1}^{3}v_t^t,
\eeq
 subject to the constraint that
 \beq
 \mathbf{M}_*^{\rm s}=\frac{1}{2}\mathbf{1},
 \eeq
 where $\mathbf{M}_*^{\rm s} \equiv \frac{1}{3}\sum_{t=1}^{3}\mathbf{M}_t^{\rm s}$.
 A high value of $C_\mathrm{M}$ signals that each of the three secondary measurements is close to the corresponding primary one. Table~\ref{tb:meas_weights} reports the weights we obtain from this optimization procedure, averaged over the 100 runs of the experiment. These weights yield $C_{\rm M} = 0.9976\pm0.0001$, again indicating the closeness of the secondary measurements to the primary ones.

This optimization defines the precise linear transformation of the rows of $D^{\rm p}$ and the linear transformation of the columns  of $D^{\rm p}$ that serve to define the secondary preparations and measurements.  By combining the operations on the rows and on the columns, we obtain from $D^{\rm p}$ a $3\times 6$ matrix, denoted $D^{\rm s}$, whose entries $s^{t'}_{t,b}$ are
\begin{equation}
\sum_{\tau=1}^{4} \sum_{\beta=0}^{1}  u^{t,b}_{\tau,\beta} \left[ \sum_{\tau'=1}^{4} v^{t'}_{\tau'} p^{\tau'}_{\tau,\beta} + v^{t'}_{\mathbf{0}} 0 + v^{t'}_{\mathbf{1}} 1 +  v^{t'}_{\lnot4} (1-p^{4}_{\tau,\beta})\right]
\end{equation}
where $t',t\in \{1,2,3\}$, $b\in \{0,1\}$.
This matrix describes the secondary preparations $P_{t,b}^{\rm s}$ and measurements $M_{t'}^{\rm s}$.
The component $s^{t'}_{t,b}$ of this matrix describes the probability of obtaining outcome 0 in measurement $M^{\rm s}_{t'}$ on preparation $P^{\rm s}_{t,b}$, that is,
\beq
s^{t'}_{t,b} \equiv p(0|M^{\rm s}_{t'},P^{\rm s}_{t,b}).
\eeq
These probabilities are the ones that are used to calculate the value of $A$ via Eq.~(6) of the main text.

\begin{table*}
\centering
\begin{tabular}{|r||rrrrrrrr|}%
\hline
&\multicolumn{1}{c}{${P}_{1,0}^{\rm p}$} & \multicolumn{1}{c}{${P}_{1,1}^{\rm p}$} & \multicolumn{1}{c}{${P}_{2,0}^{\rm p}$} & \multicolumn{1}{c}{${P}_{2,1}^{\rm p}$} & \multicolumn{1}{c}{${P}_{3,0}^{\rm p}$} & \multicolumn{1}{c}{${P}_{3,1}^{\rm p}$} & \multicolumn{1}{c}{${P}_{4,0}^{\rm p}$} & \multicolumn{1}{c|}{${P}_{4,1}^{\rm p}$} \\
\hline\hline
$P^{\rm s}_{1,0}$ & \cellcolor{gray}0.99483 & 0.00023 & 0.00029 & 0.00092 & 0.00016 & 0.00031 & 0.00324 & 0.00003 \\
$P^{\rm s}_{1,1}$ & 0.00002 & \cellcolor{gray}0.99791 & 0.00014 & 0.00026 & 0.00006 & 0.00005 & 0.00154 & 0.00002 \\
$P^{\rm s}_{2,0}$ & 0.00065 & 0.00008 & \cellcolor{gray}0.99684 & 0.00003 & 0.00001 & 0.00029 & 0.00002 & 0.00208 \\
$P^{\rm s}_{2,1}$ & 0.00134 & 0.00015 & 0.00009 & \cellcolor{gray}0.99482 & 0.00008 & 0.00028 & 0.00000 & 0.00323 \\
$P^{\rm s}_{3,0}$ & 0.00008 & 0.00023 & 0.00011 & 0.00000 & \cellcolor{gray}0.99883 & 0.00004 & 0.00044 & 0.00027 \\
$P^{\rm s}_{3,1}$ & 0.00011 & 0.00023 & 0.00022 & 0.00016 & 0.00016 & \cellcolor{gray}0.99803 & 0.00050 & 0.00061 \\
\hline
\end{tabular}
\caption{
Each of the six secondary preparation procedures, denoted $P^{\rm s}_{t,b}$ where $ t\in \{1,2,3\}, b\in \{0,1\}$ (the rows), is a probabilistic mixture of the eight primary preparation procedures, denoted  $P^{\rm p}_{t',b'}$ where $ t'\in \{1,2,3,4\}, b'\in \{0,1\}$ (the columns).  The table presents the weights appearing in each such mixture, denoted $u^{t,b}_{t',b'}$ in the main text.  These are determined numerically by maximizing the function  $C_{\rm P}=\frac{1}{6}\sum_{t=1}^3\sum_{b=0}^1 u_{t,b}^{t,b}$ (the average of the weights appearing in the shaded cells), which quantifies the closeness of the secondary procedures to the primary ones, subject to the constraint of operational equivalence of the uniform mixtures of $P_{t,0}^{\rm s}$ and $P_{t,1}^{\rm s}$ for $t\in \{1,2,3\}$.
The values presented are averages over 100 runs.
}
\label{tb:prep_weights}
\end{table*}

\begin{table*}
\centering
\begin{tabular}{|r||rrrrrrr|}
\hline
 & \multicolumn{1}{c}{$[0|{M}_1^{\rm p}]$}    & \multicolumn{1}{c}{$[0|{M}_2^{\rm p}]$}  &  \multicolumn{1}{c}{$[0|{M}_3^{\rm p}]$} &  \multicolumn{1}{c}{$[0|{M}_4^{\rm p}]$} & \multicolumn{1}{c}{${[1|M}_{4}^{\rm p}]$}  & \multicolumn{1}{c}{$1$} & \multicolumn{1}{c|}{$0$}  \\
 \hline \hline
$[0|M_1^{\rm s}]$ &\cellcolor{gray} 0.99707 & 0.00004 & 0.00015 & 0.00010 & 0.00208 & 0.00031 & 0.00025 \\
$[0|M_2^{\rm s}]$ & 0.00007 & \cellcolor{gray}0.99727 & 0.00012 & 0.00004 & 0.00199 & 0.00028 & 0.00023 \\
$[0|M_3^{\rm s}]$ & 0.00004 & 0.00002 & \cellcolor{gray}0.99845 & 0.00001 & 0.00117 & 0.00019 & 0.00012 \\
\hline
\end{tabular}
\caption{
Each of the three secondary outcome-0 measurement events, denoted $[0|M^{\rm s}_{t}]$ where $ t\in \{1,2,3\}$ (the rows), is a probabilistic mixture of the four primary outcome-0 measurement events, denoted  $[0|M^{\rm p}_{t'}]$ where $ t'\in \{1,2,3,4\}$, and three processings thereof, denoted $[1|{M}_{4}^{\rm p}]$,  $1$, and $0$  (the seven columns).  The table presents the weights appearing in each such mixture.
These are determined numerically by maximizing the function  $C_{\rm M}=\frac{1}{3}\sum_{t=1}^3 v_{t}^{t}$  (the average of the weights appearing in the shaded cells), which quantifies the closeness of the secondary procedures to the primary ones, subject to the constraint of operational equivalence between the uniform mixture of $M_1^{\rm s}$, $M_2^{\rm s}$ and $M_3^{\rm s}$ and a fair coin flip.  The values presented are averages over 100 runs. }
\label{tb:meas_weights}
\end{table*}

\section{Data analysis}
\label{suppD}

\subsection{Fitting the raw data to a generalised probabilistic theory}
\label{fittingrawdata}

In our experiment we perform four measurements on each of eight input states. If we define $r^{t'}_{t,b}$ as the fraction of `0' outcomes returned by measurement $M_{t'}$ on preparation $P_{t,b}$, the results can be summarized in a $4\times8$ matrix of raw data, $D^{\rm r}$, defined as:
\begin{equation}
D^{\rm r} =
 \begin{pmatrix}
  r^{1}_{1,0}  & r^{1}_{1,1} & \cdots & r^{1}_{4,0}  & r^{1}_{4,1} \\
  r^{2}_{1,0}  & r^{2}_{1,1} & \cdots & r^{2}_{4,0}  & r^{2}_{4,1} \\
  r^{3}_{1,0}  & r^{3}_{1,1} & \cdots & r^{3}_{4,0}  & r^{3}_{4,1} \\
  r^{4}_{1,0}  & r^{4}_{1,1} & \cdots & r^{4}_{4,0}  & r^{4}_{4,1}
 \end{pmatrix}.
\end{equation}
Each row of $D^{\rm r}$ corresponds to a measurement, ordered from top to bottom as $M_{1}$, $M_2$, $M_3$, and $M_4$. Similary, the columns are labelled from left to right as $P_{1,0}$, $P_{1,1}$, $P_{2,0}$, $P_{2,1}$, $P_{3,0}$, $P_{3,1}$,$P_{4,0}$, and $P_{4,1}$.

In order to test the assumption that three independent binary-outcome measurements are tomographically complete for our system, we fit the raw data to a matrix, $D^{\rm p}$, of primary data defined in Eq.~\eqref{prim_matrix}.
$D^{\rm p}$ contains the outcome probabilities of four measurements on eight states in the GPT-of-best-fit to the raw data. We fit to a GPT in which three 2-outcome measurements are tomographically complete, which we characterize with the following result.

\begin{proposition}\label{rawprop}
A matrix $D^p$ can arise from a GPT in which three two-outcome measurements are tomographically complete if and (with a measure zero set of exceptions) only if $a p^1_{t,b}+b p^2_{t,b}+c p^3_{t,b}+ d p^4_{t,b} -1 = 0$ for some real constants $\{a,b,c,d\}$.
\end{proposition}
\begin{proof}
We begin with the ``only if'' part. Following \cite{hardy01,barrett07}, if a set of two-outcome measurements $M_A, M_B, M_C$ (called {\em fiducial} measurements) are tomographically complete for a system, then the state of the system given a preparation $P$ can be specified by the vector
\begin{equation}
\mathbf{p} = \begin{pmatrix}1 \\ p(0|M_A, P)\\p(0|M_B,P)\\p(0|M_C,P)\end{pmatrix}
\label{pvector}
\end{equation}
(where the first entry indicates that the state is normalized). In \cite{hardy01,barrett07} it is shown that convexity then requires that the probability of outcome `0' for any measurement $M$ is given by $\mathbf{r} \cdot \mathbf{p}$ for some vector $\mathbf{r}$. Let $\mathbf{r}_1, \mathbf{r}_2, \mathbf{r}_3, \mathbf{r}_4$ correspond to outcome `0' of the measurements $M_1, M_2, M_3, M_4$, and note that the measurement event that {\em always} occurs, regardless of the preparation (e.g. the event of obtaining either outcome `0' or `1' in any binary-outcome measurement), must be represented by $\mathbf{r}_{\idn} = (1,0,0,0)$.
\color{black}
Since the ${\mathbf{r}_1, \mathbf{r}_2, \mathbf{r}_3, \mathbf{r}_4, \mathbf{r}_{\idn}}$ are a set of five four-dimensional vectors, they must be linearly dependent:
\begin{equation}
  a' \mathbf{r}_1 + b' \mathbf{r}_2 + c' \mathbf{r}_3 + d' \mathbf{r}_4 + e' \mathbf{r}_{\idn} = 0 \label{lindep1}
\end{equation}
 with $(a',b',c',d', e') \neq (0,0,0,0,0)$.
The set of $\mathbf{r}$ for which $e'$ \emph{must} be zero are those where $\mathbf{r}_{\idn}$ is not in the span of $\mathbf{r}_1, \mathbf{r}_2, \mathbf{r}_3, \mathbf{r}_4$, which is a set of measure zero. Hence we can generically ensure $e' \neq 0$ and divide Eq.~\eqref{lindep1} through by $-e'$ to obtain
\begin{equation}
  a \mathbf{r}_1 + b \mathbf{r}_2 + c \mathbf{r}_3 + d \mathbf{r}_4 - \mathbf{r}_{\idn} = 0 \label{lindep2}
\end{equation}
where $a = -a'/e'$, $b = -b'/e'$ and so on.

Finally, letting $\mathbf{p}_{t,b}$ denote the column vector of the form of Eq.~\eqref{pvector} that is associated to the preparation $P_{t,b}$, and noting that by definition
\beq
p^{t'}_{t,b} =  \mathbf{r}_{t'} \cdot \mathbf{p}_{t,b},
\eeq
we see that by taking the dot product of Eq.~\ref{lindep2} with each $\mathbf{p}_{t,b}$, we obtain the desired constraint on $D_{\rm p}$.

For the ``if'' part, we assume the constraint and demonstrate that there exists a triple of binary-outcome measurements, $M_A$, $M_B$, and $M_C$, that are tomographically complete for the GPT.   To establish this, it is sufficient to take  the fiducial set, $M_A$, $M_B$ and $M_C$, to be $M_1$, $M_2$, and $M_3$, so that preparation $P_{t,b}$ corresponds to the vector
\begin{equation}
  \mathbf{p}_{t,b} = \begin{pmatrix} 1 \\ p^1_{t,b} \\ p^2_{t,b} \\ p^3_{t,b}\end{pmatrix}.
\end{equation}
In this case, we can recover $D^{\rm p}$ if $M_1$, $M_2$, and $M_3$ are represented by $\mathbf{r}_1 = (0,1,0,0)$, $\mathbf{r}_2 = (0,0,1,0)$ and $\mathbf{r}_3 = (0,0,0,1)$, whilst the assumed constraint implies that $\mathbf{r}_4 = -(-1, a, b, c)/d$.
\end{proof}

Geometrically, the proposition dictates that the eight columns of $D^{\rm p}$ lie on the 3-dimensional hyperplane defined by the constants $\{a,b,c,d\}$.

To find the GPT-of-best-fit we fit a 3-d hyperplane to the eight 4-dimensional points that make up the columns of $D^{\rm r}$. We then map each column of $D^{\rm r}$ to its closest point on the hyperplane, and these eight points will make up the columns of $D^{\rm p}$. We use a weighted total least-squares procedure~\cite{krystek07,numrec} to perform this fit. Each element of $D^{\rm r}$ has an uncertainty, $\Delta r^{t'}_{t,b}$, which is estimated assuming the dominant source of error is the statistical error arising from Poissonian counting statistics. We define the \emph{weighted distance}, $\chi_{t,b}$, between the $(t,b)$ column of $D^{\rm r}$ and $D^{\rm p}$ as $\chi_{t,b}=\sqrt{\sum_{t'=1}^4 \left(r^{t'}_{t,b}-p^{t'}_{t,b}\right)^2/\left(\Delta r^{t'}_{t,b}\right)^2}$. Finding the best-fitting hyperplane can be summarized as the following minimization problem:
\begin{equation}
 \begin{aligned}\label{eq:chisq_comp}
 & \underset{\{p^i_{t,b},a,b,c,d\}}{\text{minimize}}
 & & \chi^2 = \sum_{t=1}^4\sum_{b=0}^1 \chi_{t,b}^2, \\
 & \text{subject to}
 & & a p^1_{t,b} + b p^2_{t,b} + c p^3_{t,b} + d p^4_{t,b} - 1 = 0 & \\
 & & & \forall \,t=1,\ldots,4, b=0,1.
 \end{aligned}
\end{equation}

The optimization problem as currently phrased is a problem in 36 variables---the 32 elements of $D^{\rm p}$ together with the hyperplane parameters $\{a,b,c,d\}$. We can simplify this by first solving the simpler problem of finding the weighted distance $\chi_{t,b}$ between the $(t,b)$ column of $D^{\rm r}$ and the hyperplane $\{a,b,c,d\}$. This can be phrased as the following 8-variable optimization problem:
\begin{equation}
 \begin{aligned}\label{eq:chi_j}
 & \underset{\{p^1_{t,b},p^2_{t,b},p^3_{t,b},p^4_{t,b}\}}{\text{minimize}} & & \chi_{t,b}^2 = \sum_{t'=1}^4 \frac{(r^{t'}_{t,b}-p^{t'}_{t,b})^2}{\left(\Delta r^{t'}_{t,b}\right)^2}, \\
 & \text{subject to} & & a p^1_{t,b} + b p^2_{t,b} + c p^3_{t,b} + d p^4_{t,b} - 1 = 0.
 \end{aligned}
\end{equation}
Using the method of Lagrange multipliers~\cite{krystek07}, we define the Lagrange function
$\Gamma = \chi_{t,b}^2 + \gamma(a p^1_{t,b} + b p^2_{t,b} + c p^3_{t,b} + d p^4_{t,b} - 1)$, where $\gamma$ denotes the Lagrange multiplier, then simultaneously solve
\begin{equation}
\frac{\partial\Gamma}{\partial\gamma}=0
\end{equation}
and
\begin{equation}
\frac{\partial\Gamma}{\partial p^{t'}_{t,b}}=0, \; t' = 1, \ldots, 4
\end{equation}
for the variables $\gamma$, $p^1_{t,b}$, $p^2_{t,b}$, $p^3_{t,b}$, and $p^4_{t,b}$. Substituting the solutions for $p^1_{t,b}$, $p^2_{t,b}$, $p^3_{t,b}$ and $p^4_{t,b}$ into Eq.~\eqref{eq:chi_j} we find
\begin{equation}
\chi_{t,b}^2 = \frac{(a r^1_{t,b} + b r^2_{t,b} + c r^3_{t,b} + d r^4_{t,b} - 1)^2}{\left(a\Delta r^1_{t,b}\right)^2 + \left(b\Delta r^2_{t,b}\right)^2 + \left(c\Delta r^3_{t,b}\right)^2 + \left(d\Delta r^4_{t,b}\right)^2},
\end{equation}
which now only contains the variables $a$, $b$, $c$, and $d$.

The hyperplane-finding problem can now be stated as the following four-variable optimization problem:
\begin{equation}
 \begin{aligned}
  \underset{\{a,b,c,d\}}{\text{minimize}}
 & & \chi^2 = \sum_{t=1}^4\sum_{b=0}^1 \chi_{t,b}^2
 \end{aligned}
\end{equation}
which we solve numerically.

The $\chi^2$ parameter returned by the fitting procedure is a measure of the goodness-of-fit of the hyperplane to the data. Since we are fitting eight datapoints to a hyperplane defined by four fitting parameters $\{a,b,c,d\}$, we expect the $\chi^2$ parameter to be drawn from a $\chi^2$ distribution with four degrees of freedom~\cite{numrec}, which has a mean of 4. As stated in the main text, we ran our experiment 100 times and obtained 100 independent $\chi^2$ parameters; these have a mean of $3.9\pm0.3$. In addition we performed a more stringent test of the fit of the model to the data by summing the counts from all 100 experimental runs before performing a single fit. This fit returns a $\chi^2$ of 4.33, which has a $p$-value of 36\%. The outcomes of these tests are consistent with our assumption that the raw data can be explained by a GPT in which three 2-outcome measurements are tomographically complete and which also exhibits Poissonian counting statistics. Had the fitting procedure returned $\chi^2$
values that were much higher, this would have indicated that the theoretical description of the preparation and measurement procedures required more than three degrees of freedom. On the other hand, had the fitting returned an average $\chi^2$ much lower than 4, this would have indicated that we had overestimated the amount of uncertainty in our data.

After finding the hyperplane-of-best-fit $\left\{a,b,c,d\right\}$, we find the points on the hyperplane that are closest to each column of $D^{\rm r}$. This is done by numerically solving for $p^1_{t,b}$, $p^2_{t,b}$, $p^3_{t,b}$, and $p^4_{t,b}$ in \eqref{eq:chi_j} for each value of $(t,b)$. The point on the hyperplane closest to the $(t,b)$ column of $D^{\rm r}$ becomes the $(t,b)$ column of $D^{\rm p}$. The matrix $D^{\rm p}$ is then used to find the secondary preparations and measurements.

\subsection{Why is fitting to a GPT necessary?}
It is clear that one needs to assume that the measurements one has performed form a tomographically complete set, otherwise statistical equivalence relative to those measurements does not imply statistical equivalence relative to all measurements. (Recall that the assumption of preparation noncontextuality only has nontrivial consequences when two preparations are statistically equivalent for all measurements.)

The minimal assumption for our experiment would therefore be that the four measurements we perform are tomographically complete. But our physical understanding of the experiment leads us to a stronger assumption, that three measurements are tomographically complete. Here we clarify why, given this latter assumption, it is necessary to carry out the step of fitting to an appropriate GPT.

It is again easier to begin by considering the case that our experiment is described by quantum theory. Let $(q^1_{t,b}, q^2_{t,b}, q^3_{t,b}, q^4_{t,b})$ denote the probability of obtaining outcome `0' in measurements $M_1$, $M_2$, $M_3$, $M_4$ on preparation $P_{t,b}$, according to quantum theory, namely $q^i_{t,b} = {\rm Tr}(E_i \rho_{t,b})$, where $E_i$ is the POVM element corresponding the the $0$ outcome of measurement $M_i$ and $\rho_{t,b}$ is the density operator for $P_{t,b}$.

Let us represent $\rho_{t,b} = \vec \sigma \cdot \vec u_{t,b}$ by a Bloch vector $\vec u_{t,b}$  and the elements $E_i = v_i^0 \idn + \vec \sigma \cdot \vec v_i$ by a ``Bloch four-vector'' $(v_i^0, \vec v_i)$. Then $q^i_{t,b} = v_i^0 + \vec u_{t,b} \cdot \vec v_i$. Since the $\vec v_i$ lie in a unit sphere, the $(q^1_{t,b}, q^2_{t,b}, q^3_{t,b}, q^4_{t,b})$ lie in the image of the sphere under the affine transformation $\vec u \mapsto (v_1^0, v_2^0, v_3^0, v_4^0) + (\vec v_1 \cdot \vec u, \vec v_2 \cdot \vec u, \vec v_3 \cdot \vec u, \vec v_4 \cdot \vec u)$, i.e. some ellipsoid, a three-dimensional shape in a four-dimensional space.

However, the relative frequencies we observe will fluctuate from $q^i_{t,b}$ in all four dimensions.
Fluctuations in the three dimensions spanned by the ``Bloch ellipsoid'' can be accomodated by using secondary preparations as described above. But fluctuations in the fourth direction are, according to quantum theory, always statistical and never systematic, and by the same token we cannot deliberately produce supplementary preparations that have any bias in this fourth direction.
Therefore, we need to deal with these fluctuations in a different way. If one was assuming quantum theory, one would simply fit relative frequencies to the closest points $q^{t'}_{t,b}$ in the Bloch ellipsoid, just as one usually fits to the closest valid density operator.

Since we do not assume quantum theory, we do not assume that the states lie in an ellipsoid. However, we still make the assumption that three two-outcome measurements are tomographically complete. Hence, by Proposition \ref{rawprop}, the long-run probabilities lie in a three-dimensional subspace of a four-dimensional space, and so there are no supplementary preparations that can deal with fluctuations of relative frequencies in the fourth dimension. Instead of fitting to the ``Bloch ellipsoid'', we fit to a suitable GPT.

\subsection{Analysis of statistical errors}

Because the relative frequencies derived from the raw data constitute a finite sample of the true probabilities (i.e. the long-run relative frequencies), the GPT states and effects that yield the best fit to the raw data are {\em estimates} of the GPT states and effects that characterize the primary preparations and measurements.

It is these estimates that we input into the linear program that identifies the weights with which the primary procedures must be mixed to yield secondary procedures.  As such, our linear program outputs estimates of the true weights, and therefore when we use these weights to mix our estimates of the GPT states and effects that characterize the primary preparations and measurements, we obtain estimates of the GPT states and effects that characterize the secondary preparations and measurements.  In turn,  these estimates are input into the expression for $A$ and yield an estimate of the value of $A$ for the secondary preparations and measurements.

To determine the statistical error on our estimate of $A$, we must quantify the statistical error on our estimates of the GPT states for the primary preparations and on our estimates of the GPT effects for the primary measurements.  We do so by taking our experimental data in 100 distinct runs, each of which yields one such estimate.  For each of these, we follow the algorithm for computing the value of $A$.  In this way, we obtain 100 samples of the value of $A$ for the secondary procedures, and these are used to determine the statistical error on our estimate for $A$.

Note that a different approach would be to presume some statistical noise model for our experiment, then input the observed relative frequencies (averaged over the entire experiment) into a program that adds noise using standard Monte Carlo techniques.  Though one could generate a greater number of samples of $A$ in this way, such an approach would be worse than the one we have adopted because the error analysis would be only as reliable as one's assumptions regarding the nature of the noise.

Given that the quantity $A$ we obtain is 2300$\sigma$ above the noncontextual bound, we can conclude that there is a very low likelihood that a noncontextual model would provide a better fit to the true probabilities than the GPT that best fit our finite sample would. This is the sense in which our experiment rules out a noncontextual model with high confidence.

It should be noted that this sort of analysis of statistical errors is no different from that used for experimental tests of Bell inequalities.  The Bell quantity (the expression that is bounded in a Bell inequality) is defined in terms of the true probabilities.  Any  Bell experiment, however, only gathers a finite sample of these true probabilities.  From this sample, one estimates the true probabilities and in turn the value of the Bell quantity.   We treat the quantity $A$ appearing in our noncontextuality inequality in a manner precisely analogous to the Bell quantity. The definition of $A$ in terms of the true probabilities is admittedly more complicated than for a Bell quantity:  we define secondary procedures based on an optimization problem that takes as input the true probabilities for the primary procedures, and use the true probabilites for the secondary procedures to define $A$. But this complication does not change the fact that $A$ is ultimately just a function of the true probabilities for the primary preparations and measurements, albeit a function that incorporates a particular linear optimization problem in its definition.\color{black}

\section{Experimental methods}
A 20-mW diode laser with a wavelength of $404.7$~nm produces photon pairs, one horizontally polarized the other vertically polarized, via spontaneous parametric down-conversion in a $20$-mm type-II PPKTP crystal. The downconversion crystal is inside a Sagnac loop and the pump laser is polarized vertically to ensure it only travels counter-clockwise around the loop. Photon pairs are separated at a polarizing beamsplitter and coupled into two single-mode fibres (SMFs). Vertically-polarized photons are detected immediately at detector $D_h$, heralding the presence of the horizontally-polarized signal photons which emerge from SMF and pass through a state-preparation stage before they are measured. Herald photons were detected at a rate of $400\,\mathrm{kHz}$. The single photon detection rate at detectors $D_r$ and $D_t$ depends on the measurement settings. In the transmissive and reflective ports of the Glan-Taylor PBS (GT-PBS) used in the measurement, photons were detected at maximum rates of $330\,\mathrm{kHz}$ and $250\,\mathrm{kHz}$, respectively. Coincident detection events between herald photons and the transmissive and reflective ports of the measurement PBS were up to $22\,\mathrm{kHz}$ and $16\,\mathrm{kHz}$, respectively.

Signal photons emerge from the fibre and pass through a Glan-Taylor PBS which transmits vertically polarised light. Polarization controllers in the fibre maximize the number of photons which pass through the beamsplitter. A quarter- and half-waveplate set the polarization of the signal photons to one of eight states.

An SMF acts as a spatial mode filter. This filter ensures that information about the angles of the state-preparation waveplates cannot be encoded in the spatial mode of the photons, and that our measurement procedures do not have a response that depends on the spatial mode,
but only on polarization as intended.
The SMF induces a fixed polarization rotation, so a set of three compensation waveplates are included after the SMF to undo this rotation.
It follows that the preparation-measurement pairs implemented in our experiment are in fact a rotated version of the ideal preparation and a similarly-rotated version of the ideal measurement.  Such a fixed rotation, however, does not impact any of our analysis.

Measurements are performed in four bases, set by a half- and quarter-waveplate. A second Glan-Taylor PBS splits the light, and both output ports are detected. Due to differences in the coupling and detection efficiencies in each path after the beamsplitter, each measurement consists of two parts. First, the waveplates are aligned such that states corresponding to outcome `0' are transmitted by the PBS, and the number of heralded photons detected in a two-second window is recorded for each port. Second, the waveplate angles are changed in such a way as to invert the outcomes, so the detector in the reflected port corresponds to outcome `0' and heralded photons are detected for another two seconds. The counts are added together and the probability for outcome `0' is calculated by dividing the number of detections corresponding to outcome `0' by the total number of detection events in the four-second window.